\documentclass[11pt]{article}
\usepackage[margin=1in]{geometry}
\usepackage{times}

\usepackage{float}
\usepackage{amsfonts}
\usepackage{amsmath,amsthm,amssymb}
\usepackage{latexsym}
\usepackage{epic}
\usepackage{epsfig}
\usepackage{hyperref}
\usepackage{multirow}
\usepackage{hhline}
\usepackage{verbatim}
\usepackage[justification=centering]{caption}
\usepackage{enumitem}
\usepackage{color}
\allowdisplaybreaks[1]
\usepackage[numbers]{natbib}
\usepackage{comment}
\usepackage{algorithmic}
\usepackage[linewidth=.5pt]{mdframed}
\usepackage{lipsum}
\usepackage{booktabs} 

\usepackage[linesnumbered,ruled,vlined]{algorithm2e} 

\SetCommentSty{mycommfont}

\SetAlFnt{\small}
\SetAlCapFnt{\small}
\SetAlCapNameFnt{\small}
\SetAlCapHSkip{0pt}
\IncMargin{-\parindent}

\newtheorem{theorem}{Theorem}[section]

\newtheorem{lemma}[theorem]{Lemma}
\newtheorem{proposition}[theorem]{Proposition}

\newtheorem{definition}[theorem]{Definition}
\newtheorem{remark}[theorem]{Remark}

\newtheorem{fact}[theorem]{Fact}
\newtheorem{example}[theorem]{Example}

\newtheorem*{conjecture*}{Conjecture}
\newtheoremstyle{nonindented}{1ex}{1ex}{}{}{\bfseries}{.}{.5em}{}
\newtheoremstyle{indented}{1ex}{1ex}{\itshape\addtolength{\leftskip}{0.6cm}\addtolength{\rightskip}{0.6cm}}{}{\bfseries}{.}{.5em}{}
\theoremstyle{nonindented}
\theoremstyle{indented}
\theoremstyle{plain}

\newcommand{\set}[1]{\left\{ #1 \right\}}

\renewcommand{\hat}{\widehat}
\renewcommand{\tilde}{\widetilde}
\renewcommand{\bar}{\overline}

\DeclareMathOperator{\poly}{poly}

\def\min{\qopname\relax n{min}}
\def\max{\qopname\relax n{max}}
\def\maxs{\qopname\relax n{max2}}

\def\argmax{\qopname\relax n{argmax}}

\def\Pr{\qopname\relax n{\mathbf{Pr}}}
\def\Ex{\qopname\relax n{\mathbf{E}}}

\newcommand{\RR}{\mathbb{R}}

\def\F{\mathcal{F}}

\def\I{\mathcal{I}}

\def\P{\mathcal{P}}

\def\S{\mathcal{S}}
\def\O{\mathcal{O}}

\def\eps{\epsilon}
\def\sse{\subseteq}

\newcommand{\mini}[1]{\mbox{minimize} & {#1} &\\}
\newcommand{\maxi}[1]{\mbox{maximize} & {#1 } & \\}

\newcommand{\st}{\mbox{subject to} }
\newcommand{\con}[1]{&#1 & \\}
\newcommand{\qcon}[2]{&#1, & \mbox{for } #2.  \\}
\newenvironment{lp}{\begin{equation}  \begin{array}{lll}}{\end{array}\end{equation}}
\newenvironment{lp*}{\begin{equation*}  \begin{array}{lll}}{\end{array}\end{equation*}}

\begin{document}
\title{On the Tractability of Public Persuasion with No Externalities} 
\author{Haifeng Xu  \\ 
	University of Virginia \\
		{\tt hx4ad@virginia.edu}}

\date{}

\begin{titlepage}
	\clearpage\maketitle
	\thispagestyle{empty}

\begin{abstract}
	
\emph{Persuasion} studies how a principal can influence agents' decisions  via strategic information revelation --- often described as a \emph{signaling scheme} ---  in order to yield the most desirable equilibrium outcome. 
A basic question that has attracted much recent attention is how to compute the optimal public signaling scheme, a.k.a., \emph{public persuasion}, which is motivated by various applications including auction design, routing, voting, marketing, queuing, etc. Unfortunately, most algorithmic studies  in this space exhibit quite negative results and are rifle with computational intractability.  Given such background, this paper seeks to understand when public persuasion is  tractable and how tractable it can be. We focus on
a fundamental multi-agent persuasion model introduced by Arieli and Babichenko \cite{Arieli2016}: many agents, \emph{no} inter-agent externalities and binary agent actions, and identify well-motivated circumstances under which efficient algorithms are possible.   
En route, we also develop new algorithmic techniques and demonstrate that they can be  applicable to other public persuasion problems or even beyond. 

We start by proving that optimal public persuasion in our model is fixed parameter tractable. Our main result here builds on an interesting connection to a basic question in  combinatorial geometry: \emph{how many cells can $n$ hyperplanes divide $\RR^d$ into?} We use this connection to show a new characterization of public persuasion, which then enables efficient algorithm design. Second, we relax agent incentives and show that  optimal public persuasion admits a bi-criteria PTAS  for the widely studied class of monotone submodular objectives, and this approximation  is tight. To prove this result, we  establish an intriguing ``noise stability'' property of submodular functions which strictly generalizes the key result of Cheraghchi \emph{et al.} \cite{CKKL12}, originally motivated by applications of learning submodular functions and differential privacy. Finally,  motivated by automated application of persuasion, we consider relaxing the equilibrium concept of the model to coarse correlated equilibrium. Here, using a sophisticated primal-dual analysis, we prove that optimal public persuasion admits an efficient algorithm \emph{if and only if } the combinatorial problem of maximizing the sender's objective minus any linear function can be solved efficiently, thus establishing their polynomial-time equivalence.

\end{abstract}

\end{titlepage}

\newpage

\section{Introduction}
The study of how beliefs influence agents' decisions has a rich history and forms the celebrated field of information economics. Along this line, there has been a recent surge of interest  in understanding how a principal can strategically shape agents' beliefs  in order to steer their collective decisions towards the most desirable equilibrium outcome. This task ---  often referred to as \emph{persuasion} or \emph{signaling} ---  is intrinsically an optimization problem, and  has attracted much  attention in algorithmic game theory  \cite{mixture_selection,Dughmi2016,babichenko2016,Dughmi2017algorithmic,Dughmi2018hardness}. Such interest in algorithms is partially driven by the potential automated application of persuasion in domains such as  auctions \cite{Emek12,Miltersen12}, traffic routing \cite{Bhaskar2016}, recommendation systems \cite{Mansour2016bayesian}, marketing \cite{babichenko2016,Ozan19}, customer queuing \cite{Lingenbrink2019},  voting \cite{Alonso14,mixture_selection}, and security \cite{Xu15,Rabinovich15}.

 A foundational model that illustrates the essence of persuasion is the Bayesian persuasion (BP) model of Kamenica and Gentzkow \cite{Kamenica2011}. Here there are two players, a  \emph{sender} (she) and a \emph{receiver} (he). The receiver is faced with selecting one of a number of actions, the payoffs of which are uncertain and depend on a \emph{state of nature} drawn from a prior distribution known to both players. However, the sender possesses an informational advantage and can observe the \emph{realized} state of nature. 
In order to improve her utility, the sender can commit to a policy --- often known as a \emph{signaling scheme} ---  to strategically reveal her information before the receiver selects his action. 
Optimal persuasion refers to the sender's task of designing the signaling scheme to maximize her utility.  

In many applications, the sender faces \emph{multiple receivers}  and can reveal information to each of them. 
In such settings, there are generally two basic signaling models. The first is  \emph{private  persuasion} where the sender can reveal different information to different receivers through a private communication channel. The second  model --- which is the focus of this paper --- is \emph{public persuasion} where the sender  reveals the same information to each receiver via a \emph{public signaling scheme}.  Though private persuasion  generally yields higher sender utility, the study of public persuasion has been the focus in much of previous literature (e.g.,   \cite{Emek12,Miltersen12,Alonso14,mixture_selection,Bhaskar2016,lingenbrink2018signaling,Dughmi2018hardness,Ozan19}). This is due to several reasons. First, in many settings, there are too many receivers and privately communicating with each receiver is either too costly or impractical (e.g., persuading a large population of voters). Second, private persuasion assumes that receivers do \emph{not} share their signals/information with each other, which does not hold in many applications. Third, sometimes revealing disparate information is undesirable due to concerns such as unfairness.

In a seminal work, Kamenica and Gentzkow characterize optimal public persuasion as the \emph{concavification} of the sender's objective function \cite{Kamenica2011}. Interestingly, despite this simple mathematical characterization, the corresponding algorithmic problem of efficiently computing the optimal public scheme appears quite difficult. 
This is due to at least two key dimensions of challenges in public persuasion, described as follows.   
\begin{enumerate}\vspace{-1mm}
	\item \emph{Externalities among receivers}, as evidenced by previous hardness results for perhaps the most basic setting in this space. Specifically, \cite{Bhaskar2016,rubinstein2015eth,Dughmi2018hardness} consider  a sender who looks to persuade two receivers playing a Bayesian \emph{zero-sum} game, and rule out any  Polynomial Time Approximation Scheme (PTAS) for optimal public persuasion assuming planted-clique hardness or the Exponential Time Hypothesis (ETH). A QPTAS was provided in \cite{mixture_selection}, which essentially matches the complexity lower bound. 
	\vspace{-1mm}
	
	\item \emph{Coordinating the decisions of many receivers}.  This is illustrated by the hardness results for another basic model introduced by Arieli and Babichenko \cite{Arieli2016}: many receivers with \emph{no inter-agent externalities} and binary receiver actions. Dughmi and Xu  \cite{Dughmi2017algorithmic} prove that in this basic setting it is NP-hard to obtain a PTAS for optimal public persuasion even for linear sender objectives. So far no tractable  settings or efficient approximate algorithms  are known for this model. 
\end{enumerate}

Given these discouraging messages, the main \emph{conceptual motivation} of this paper is to understand when public persuasion is tractable and how tractable it can be.  Concretely, we focus on tackling the \emph{second} challenge above --- that of coordinating the decisions of receivers ---  in the fundamental multi-agent persuasion model by Arieli and Babichenko  with \emph{no} inter-receiver externalities. Each receiver takes a binary action from $\{ 0 ,1  \}$,  so the sender's utility function  is a \emph{set function} over receivers. This model nicely serves our purpose because it ``disentangles'' the complexity of  coordinating receivers from the complexity due to agent externalities, and thus allows us to focus on the former. 
Practically, the model also finds application in domains such as marketing \cite{babichenko2016},  voting \cite{Arieli2016}, influence on networks \cite{Ozan19} and customer queuing \cite{Lingenbrink2019}, to name a few.  We remark that there is no lack of algorithmic study of this model, however previous works all primarily focused on \emph{private} persuasion \cite{Arieli2016,babichenko2016,Dughmi2017algorithmic}.  
In contrast, we design efficient and tight algorithms for \emph{public} persuasion in well-motivated circumstances.  En route, we also develop new algorithmic techniques and demonstrate that they are applicable to other public persuasion problems or even beyond.

 \subsection{Our Results and Techniques}
We first prove that public persuasion in our model  is \emph{fix-parameter tractable}.  Our main result here shows that  the optimal public signaling scheme can be computed in polynomial time for \emph{arbitrary sender objective} (a set function) --- even those which are intractable to optimize directly --- when the number of states of nature is a constant. This is a surprise to us since  the restriction to a small state space does \emph{not} appear to simplify the problem at the first glance. Indeed, Babichenko and Barman \cite{babichenko2016} prove that optimal private persuasion is APX-hard even in the case with only \emph{two} states of nature and monotone submodular sender objectives.  Our algorithm here is based on an interesting connection between public persuasion and a basic question in combinatorial geometry: \emph{how many cells can $n$ hyperplanes  divide the space of $\RR^d$ into?} We show that a constructive answer to this question gives rise to a new characterization of public schemes, which leads to the design of an efficient algorithm. To illustrate the power of this technique, we also show its applicability to another widely studied persuasion problem, yielding new algorithmic result for that model as well. 

Next, we consider public persuasion with \emph{relaxed receiver incentives}. Here our results concern the setting where the sender objective is state-independent. This important special case is also the focus of many previous works including the original model of Arieli and Babichenko \cite{Arieli2016}, and is realistic in various applications including voting \cite{Alonso14}, marketing \cite{babichenko2016,Ozan19} and customer queuing \cite{Lingenbrink2019} to name a few.       
We exhibit a PTAS for public persuasion with a \emph{bi-criteria} guarantee for  \emph{monotone submodular} sender objectives. This is  essentially the best possible, as we  show that there is no bi-criteria FPTAS neither single-criteria PTAS, unless P=NP.   
We also illustrate an interesting contrast between public and private persuasion by proving that  it becomes NP-hard   to obtain a bi-criteria PTAS for private persuasion in the same setting. Notably, our algorithm also works for \emph{non-monotone} submodular objectives, but with a slightly worse guarantee for the sender's expected utility.    
Our result here is built upon an intriguing property of submodular functions which we believe may be of independent interest. In particular, we prove that the evaluation of submodular functions is \emph{noise-stable}: to evaluate a submodular function $f: 2^{[n]} \to \RR_+$ at any set $S \subseteq [n]$, if one slightly perturbs set $S$ to generate a  random set $T$ by adding to $S$ or deleting from $S$  any $i \in [n]$ with probability at most $\epsilon$ in an \emph{arbitrarily correlated} manner, the expected function value $\Ex_{T} f(T)$ will not decrease much. Formally, we prove $\Ex_{T} f(T) \geq (1-2 \epsilon)f(S) $ for submodular functions whereas   $\Ex_{T} f(T)$ may decrease to $(1-n\epsilon)f(S)$ in general. This ``noise stability'' property of submodular functions strictly generalizes the key result of Cheraghchi \emph{et al.} \cite{CKKL12}, who also proved noise stability of submodular functions but under a weaker notion of noise with  only \emph{independent} random perturbations (see Section \ref{sec:bicriteria:discuss} for more discussions).    




Finally, we consider public persuasion under \emph{relaxed equilibrium} conditions.  Classical persuasion models assume that receivers play a Bayesian Nash equilibrium in any public signaling scheme, which is also a Bayes correlated equilibrium \cite{Bergemann16Bayes} in our model due to the absence of externalities. Motivated by automated applications of persuasion schemes implemented as software (e.g., recommendation systems), we relax this equilibrium concept to the coarse correlated equilibrium\footnote{Coarse correlated equilibrium for Bayesian games has been studied in other settings such as auctions \cite{Cai2014,Caragiannis2015bounding}, congestion games \cite{roughgarden2015intrinsic} and general Bayesian games \cite{Hartline2015no}, and was coined \emph{Bayesian coarse correlated equilibrium} by Hartline et al. \cite{Hartline2015no}. }  and assume that each receiver decides to either follow the signaling scheme (i.e., adopting the software) or act based to his prior belief (i.e., abandoning the software).   Under this relaxation, we prove that optimal public persuasion admits an efficient algorithm \emph{if and only if} the combinatorial problem of maximizing the sender's objective function minus any linear function (subject to no constraints) can be efficiently solved, establishing their polynomial-time equivalence. Our proof uses a sophisticated primal-dual analysis, generalizes previous techniques in \cite{Dughmi2017algorithmic} to get rid of their assumption on the monotonicity of the objective, and thus applies to a broader class of objective functions.  

\section{The Model and Preliminaries}
 \subsection{Basic Setup} 

We are a  \emph{sender} facing $n$ \emph{receivers}, denoted by set  $[n] = \{1,\cdots, n \}$.  Each receiver has two actions, denoted as action $0$ and $1$. The receiver's payoff depends only on his own action and a random \emph{state of nature} $\theta$ supported on $\Theta$. In particular, let  $u_i(\theta,1)$ and $u_{i}(\theta,0)$ denote receiver $i$'s utility for action $1$ and action $0$, respectively, in state $\theta$;  as shorthand, we use $u_i(\theta) = u_i(\theta,1) - u_{i}(\theta,0)$ to denote how much receiver $i$ prefers action $1$ over action $0$ given state of nature $\theta$. Let $u_i\in \RR^{\Theta}$ denote the vector containing $u_i(\theta)$ for all $\theta$ and call it the \emph{payoff vector} of receiver $i$.  
The sender's utility  is a function of all the receivers' actions and the state of nature $\theta$.  Let $f_{\theta}(S)$ denote the sender's utility  when the state of nature is $\theta$ and $S$ is the set of receivers who choose action $1$. The following two properties of set functions will be useful. 
\begin{equation*}
\begin{array}{l}
\text{Submodularity:} \quad  f: 2^{[n]} \to \RR \text{ is submodular if }\forall S, T \subseteq[n], f(S) + f(T) \geq f(S \cup T) + f(S \cap T) ; \\
\text{Monotonicity:} \quad \, \, \, f: 2^{[n]} \to \RR \text{ is monotone non-decreasing if }\forall T \subset S \subseteq[n], f(S ) \geq  f(T). 
\end{array}
\end{equation*} 

We remark that previous studies on this model all focus on \emph{monotone non-decreasing} sender objectives \cite{Arieli2016,babichenko2016,Dughmi2017algorithmic}. However,  we do not pose such restriction in this work.  

In Bayesian persuasion, it is assumed that $\theta$ is drawn from a common prior distribution $\lambda$, which is known to the sender and all receivers. However, the sender possesses an informational advantage, namely, access to the realized state of nature $\theta$. The sender can  \emph{commit} to a policy--- termed a \emph{signaling scheme} --- that  maps the realized $\theta$   to a \emph{signal} for the receivers. The signaling scheme may be randomized, and hence reveals noisy information regarding  $\theta$.
The order of events is as follows: (1) The sender commits to a signaling scheme $\pi$; (2) Nature draws $\theta \sim \lambda$; (3) Signals are drawn according to $\pi(\theta)$ and sent to receivers; (4) Each receive updates his belief about the state of nature and  selects their actions. Persuasion is the problem faced by the sender  who seeks to design a signaling scheme to maximize her utility.  

In the literature, two basic signaling models have been studied in   multi-receiver persuasion: (1) the sender sends a public signal $\sigma$ and every receiver learns the same information, which is referred to as \emph{public persuasion} \cite{Emek12,Bhaskar2016,Dughmi2018hardness}; (2) the sender can send different (possibly correlated) signals to different receivers privately, which is referred to as \emph{private persuasion} \cite{Arieli2016,babichenko2016,Dughmi2017algorithmic}. This work focuses on  \emph{public persuasion}, which is less studied for this model in the previous literature. 



\subsection{Public Persuasion and Equilibrium Concepts}
\label{prelim:public}
This subsection describes the model of public signaling, equilibrium concepts, and (limited) previous results on public persuasion.  Though we will not explicitly define private signaling schemes, they can be viewed as a generalization of public schemes, in which the sender just sends different signals to different receivers.  

A \emph{public signaling scheme} $\pi$ is a randomized map from $\Theta$ to the set of signals $\Sigma$. Let $\pi(\theta, \sigma)$ denote the probability of selecting signal $\sigma \in \Sigma$ at the state of nature $\theta$. Since $\pi$ is public knowledge, after receiving $\sigma$, each receiver can update his posterior belief about $\theta$ and then chooses the optimal action from $\{ 0,1 \}$ based on this posterior belief.  Throughout the paper, we assume ties are broken in favor of the sender.  

As shown in \cite{Kamenica2011,Arieli2016} via a revelation-principle style argument, there always exists an optimal public signaling scheme which is  \emph{direct} and \emph{persuasive}. By \emph{direct} we mean that signals correspond to a profile of actions $\mathbf{s} \in \{ 0 , 1\}^n$, where the $i$'th entry $s_i$ corresponds to an action recommendation to receiver $i$.  A direct scheme is \emph{persuasive} if its recommendation to each receiver is indeed a best response for him.\footnote{Persuasiveness has also been called \emph{incentive compatibility} or \emph{obedience} in prior work.} Equivalently, any $\mathbf{s} \in \{ 0 ,1 \}^n$ can be equivalently viewed as a subset $S \subseteq [n]$, containing all entries of value $1$. We will use these two notations interchangeably throughout the paper.   

In this paper, we focus (without loss) on designing direct signaling schemes. A direct  scheme can be captured by variables $ \{ \pi(\theta, S) \}_{ \theta \in \Theta, S \subseteq [n]}$, in which $\pi(\theta, S)$ is the probability of sending signal $S$ ---  i.e., publicly recommending action $1$ to receivers in set $S$ and action $0$ to receivers in $[n] \setminus S$ ---  conditioned on the state of nature $\theta$. Upon receiving signal $S$, each receiver  infers that the state of nature is $\theta$ with probability  $ \frac{1}{\sum_{\theta } \lambda_{\theta}  \pi(\theta, S)} \lambda_{\theta}  \pi(\theta, S), \forall \theta$.  Therefore, the persuasiveness for signal $S$ implies the following constraints: 
\begin{equation}
\begin{array}{c}
\frac{1}{\sum_{\theta } \lambda_{\theta}  \pi(\theta, S) }  \sum_{\theta \in \Theta}  \lambda_{\theta}  \pi(\theta, S) u_i(\theta, 1) \geq   \frac{1}{\sum_{\theta } \lambda_{\theta}  \pi(\theta, S) }  \sum_{\theta \in \Theta}  \lambda_{\theta}   \pi(\theta, S) u_i(\theta, 0), \qquad \forall i \in S \\
\frac{1}{\sum_{\theta } \lambda_{\theta}  \pi(\theta, S) }  \sum_{\theta \in \Theta}  \lambda_{\theta}  \pi(\theta, S) u_i(\theta, 1) \leq   \frac{1}{\sum_{\theta } \lambda_{\theta}  \pi(\theta, S) }  \sum_{\theta \in \Theta}  \lambda_{\theta}   \pi(\theta, S) u_i(\theta, 0), \qquad \forall i \not \in S \\
\end{array}
\end{equation}
or equivalently, $\sum_{\theta \in \Theta} \lambda_{\theta}  \pi(\theta, S)  u_i(\theta)\geq 0$ for $i \in S$ and  $\sum_{\theta \in \Theta} \lambda_{\theta}  \pi(\theta, S)  u_i(\theta)\leq 0$ for $i \not \in S$, where $u_i(\theta) = u_i(\theta, 1) - u_i(\theta, 0)$.  As a result,  we can encode  the sender's optimization problem of computing the optimal public scheme using the following linear program, the size of which is \emph{exponential} in $n$. 
\begin{lp}\label{lp:optPub}
	\maxi{\sum_{\theta \in \Theta} \lambda(\theta) \sum_{S\subseteq [n]} \pi(\theta,S)  f_{\theta}(S) }
	\st 
	\qcon{\sum_{\theta \in \Theta}\lambda(\theta) \pi(\theta,S) \cdot  u_i(\theta) \geq 0}{S \sse [n] \mbox{ and } i \in S}
	\qcon{\sum_{\theta \in \Theta}\lambda(\theta) \pi(\theta,S) \cdot  u_i(\theta) \leq 0}{S \sse [n] \mbox{ and } i \not \in S}
	\qcon{\sum_{S \subseteq [n]} \pi(\theta,S) = 1}{\theta \in \Theta}
	\qcon{ \pi(\theta,S) \geq 0}{\theta \in \Theta; S \subseteq [n]}
\end{lp}

To our knowledge, little algorithmic results are known previously regarding optimal public persuasion for this model, besides the following hardness results.
\begin{theorem}[Intractability of Public Persuasion \cite{Dughmi2017algorithmic}]
Consider public persuasion with $f_{\theta}(S) = |S|/n$ for any $\theta \in \Theta$. It is NP-hard to approximate the optimal sender utility to within any constant multiplicative factor. Moreover,
there is no additive PTAS for evaluating the optimal sender utility, unless P = NP.
\end{theorem}



\noindent {\bf Equilibrium Concepts.}  The above derivations all assumed that conditioning on any public signal receivers play a Bayesian Nash equilibrium, which is also a Bayes correlated equilibrium \cite{Bergemann16Bayes} in this setting because the receivers have no externalities.  In Section \ref{sec:bicriteria}, we will relax the incentive/persuasiveness constraints and assume that a receiver will take the recommended action so long as it is at most $\epsilon$ worse than the other action. Formally, we say a public scheme is $\epsilon$-\emph{persuasive} if the following hold for any signal $S \subseteq [n]$: 
\begin{equation}\label{eq:eps-persuasive}
\begin{array}{c}
\sum_{\theta \in \Theta}\lambda(\theta) \pi(\theta,S) \cdot  u_i(\theta) \geq -\epsilon, \, \,  \forall i \in S \quad \text{ and } \quad  \sum_{\theta \in \Theta}\lambda(\theta) \pi(\theta,S) \cdot  u_i(\theta) \leq \epsilon, \, \,  \forall  i \not \in S. 
\end{array}
\end{equation} 
Note that whenever $\epsilon$-persuasiveness is considered, we shall assume $u_i(\theta) \in [-1,1]$ by convention.

 In Section \ref{sec:relax}, we will relax the equilibrium concept to coarse correlated equilibrium (CCE). That is, a receiver will follow the scheme's recommendation so long as obedience is better off  than opting out of the signaling scheme and acting just according to his prior belief. 
This solution concept for Bayesian games is sometimes also called  \emph{Bayesian coarse correlated equilibrium}  \cite{Cai2014,Caragiannis2015bounding,Hartline2015no,roughgarden2015intrinsic}.  For distinction, we say the scheme is \emph{cce-persuasive} in this case, defined as follows.   

\begin{definition}[CCE-Persuasiveness]\label{def:cce}
	A signaling scheme $\pi$ is  \emph{cce-persuasive}  if for any receiver $i$, the expected utility of following $ \pi$ is at least $i$'s maximum  utility under the prior belief $\lambda$. Formally, 
	\begin{equation}\label{eq:cce-constraint}
	\begin{array}{l}
	\sum_{S: i \in S} \sum_{\theta \in \Theta}   \lambda_{\theta} \pi(\theta, S) u_i(\theta)     \geq \max\{ \sum_{\theta \in \Theta} \lambda_{\theta} u_i (\theta), 0  \}, \qquad \forall i \in [n]. 
	\end{array} 
	\end{equation}  
\end{definition}

Therefore, the optimal cce-persuasive signaling scheme can be computed via an exponentially large LP similar to  LP \eqref{lp:optPub}, but with cce-persuasiveness constraints described in Inequality \eqref{eq:cce-constraint}.  

We remark that the original cce-persuasiveness constraint should have been the following inequality:
$$
\begin{array}{l}
\sum_{\theta \in \Theta}   \lambda_{\theta} u_i(\theta, 1) \cdot x_{\theta,i }+  \sum_{\theta \in \Theta}   \lambda_{\theta} 
u_i(\theta, 0) \cdot (1-x_{\theta,i})    \geq \max\{ \sum_{\theta \in \Theta} \lambda_{\theta} u_i (\theta, 1), \sum_{\theta \in \Theta} \lambda_{\theta} u_i (\theta,0)  \}
\end{array}
$$
where: (1) $x_{\theta,i } = \sum_{S: i \in S} \pi(\theta, S)$ is the marginal probability of recommending action $1$ to receiver $i$ in state $\theta$; (2) the left hand side is the expected utility of receiver $i$ when following all the recommendations of $\pi$; (3) the right hand side is the maximum utility under prior belief $\lambda$.  It is easy to verify that we will arrive at Inequality \eqref{eq:cce-constraint} after subtracting $\sum_{\theta} \lambda_{\theta} u_i(\theta, 0)$ from both sides of the above inequality.  

\subsection{Input Models, Computation, and Approximation}\label{prelim:approx}
Throughout the paper, we assume \emph{value oracle} access to the sender's objectives, which are set function $f_{\theta}$'s. That is, we can query the value of $f_{\theta}(S)$ for any set $S$ using a unit of time. Other than the sender objectives, the prior distribution $\lambda$ and receivers' payoffs are all explicitly given. We will consider approximately optimal signaling schemes. 
 For convenience in stating our approximation guarantees, we always assume  $f_{\theta}$'s are \emph{non-negative} functions, i.e., $f_{\theta}(S) \geq 0$.  When a signaling scheme  yields expected sender utility at least $c$ fraction of the best possible, we say it is \emph{$c$-approximate}. Sometimes (though rarely) we also consider additive loss to the sender utility, and say the scheme is \emph{$\epsilon$-optimal} if its expected sender utility is at most $\epsilon$ less than the best possible.  
When a signaling scheme is both $\epsilon$-persuasive and $c$-approximate, we say it is a \emph{bi-criteria} approximation to emphasize its loss in both optimality and persuasiveness. 


\section{Fixed Parameter Tractability of Optimal Public Persuasion} \label{sec:constant}
In this section, we show that optimal public persuasion is \emph{fixed parameter tractable}. Note that when the number of receivers $n$ is a small constant, it is easy to see that public persuasion can be solved in polynomial time  because LP \eqref{lp:optPub} then has polynomial size. More challenging is the  setting where  the number of states of nature $|\Theta|$ is a constant. Indeed, the restriction to a small $|\Theta|$ does not appear to  simplify the problem at the first glance --- private persuasion  is proved to be APX-hard even when there are only \emph{two} states of nature and the sender's objective functions are  monotone submodular \cite{babichenko2016}.  
Surprisingly, we prove that, under mild non-degeneracy assumptions, the optimal public signaling scheme can be computed efficiently for \emph{arbitrary sender objectives} --- even those which are \emph{intractable} to optimize directly --- when $|\Theta|$ is a constant.  
This result illustrates an interesting contrast between public and private persuasion. 

The proof of our main result is based on a constructive version of a very basic question in combinatorial geometry: \emph{how many cells can $n$ hyperplanes  divide the space of $\RR^d$ into?}\footnote{Throughout, by ``cells'' we mean non-degenerated regions with non-zero volumes in $\RR^d$.}   The answer is $\O(n^d)$, which is polynomial in $n$ when $d$ is a constant.  We utilize this upper bound and, additionally, design an algorithm to efficiently identify all these cells (represented as intersection of half-spaces) for constant $d$. Interestingly, we show that these cells generated by hyperplanes give rise to a characterization of public signaling schemes, which is crucial for our algorithm design.  To illustrate the power of this approach, in Appendix  \ref{sec:constant:applictions} we also show how the same technique can be used to design a new polynomial-time algorithm for  another widely studied multi-agent persuasion problem ---  i.e., public signaling in second price auctions for revenue maximization \cite{Emek12,Miltersen12,mixture_selection} --- when the number of states there is a constant.  This also complements previous complexity results from \cite{Emek12} and completed the picture of the fixed parameter tractability of this problem.\footnote{Previously, it is only known that this problem is NP-hard when there are $n \geq 3$ bidders and admits a polynomial time when the number of bidders' value types  is a constant \cite{Emek12}.}   

We first observe that if without any assumption, public persuasion \emph{cannot} be ``easier'' than directly optimizing the sender's objective. Consider the following example. 

\begin{example}\label{ex:degeneracy}
	There are two states of nature $\theta_1, \theta_2$, each occurring with equal prior probability $1/2$. Receiver $i$'s payoff vector $u_i = ( u_i(\theta_1) , u_i(\theta_2))  =(1, -1)$ for any $i\in [n]$.  It is easy to verify that the optimal public persuasion reveals no information in this case and is equivalent to solving $\arg \max_{S \subseteq [n]} f_{\theta}(S)$. 
	\end{example} 

It turns out that the above difficulty is due to certain degeneracy in the receiver utilities. Specifically, we say that \emph{the receiver payoffs are non-degenerate} if for any subset $S$ including $|\Theta|$ receivers, their payoff vectors $\{ u_i \}_{i \in S}$ are \emph{linearly independent}. This is a minor requirement --- e.g., any payoffs perturbed by small random noise will be non-degenerate with probability $1$.  We prove the following theorem.



\begin{theorem}\label{thm:const-poly}
	Let $|\Theta| = d$ and suppose receiver payoffs are non-degenerate.  There is a $\poly(n^d)$ time algorithm that computes the optimal public signaling scheme for \emph{arbitrary} sender objective functions  $\{ f_{\theta} \}_{\theta \in \Theta}$.
\end{theorem}

\noindent{\bf Proof Sketch. } We provide a sketch here and defer the full proof to Appendix \ref{sec:const:app}.  The main difficulty in computing the optimal public scheme is that there are $2^n$  different public signals. Any efficient algorithm simply cannot search over all these  signals. Our key insight is that when $|\Theta| = d$ is small, most of these public signals actually will never arise in \emph{any} signaling scheme, assuming non-degeneracy. This turns out to be a consequence of  the problem of dividing $\RR^d$ by hyperplanes. 

Specifically, each public signal induces a posterior distribution $p \in \Delta_d$ over the states in  $\Theta$ where $p_{\theta}$ denotes the probability of $\theta \in \Theta$. The hyperplane $\sum_{\theta \in \Theta} u_i(\theta)p_{\theta}  =  0$ cut $\RR^d$  into two cells: $\sum_{\theta \in \Theta} u_i(\theta)p_{\theta}  > 0$ where receiver $i$ always prefers action $1$ and  $\sum_{\theta \in \Theta} u_i(\theta)p_{\theta}  < 0$ where receiver $i$ always prefers action $0$. With $n$ receivers,  $\RR^d$ will be cut by $n$ hyperplanes into $\mathcal{O}(n^d)$ cells and each cell is uniquely characterized by a $n$-dimensional \emph{binary vector} $\mathbf{s} \in \{ 0,1 \}^n$, in which the $i$'th entry $s_i$ indicates receiver $i$'s best response action.   We  call vector $\mathbf{s}$ the \emph{label} of that cell. Any $p$ from a cell of label $\mathbf{s}$ can induce $s_i$ as receiver $i$'s best response. In other words, each label $\mathbf{s}$  corresponds to a public signal that can possibly arise. 

We wish to show that these $\mathcal{O}(n^d)$ labels of cells are  precisely all the public signals that can possibly arise in public persuasion. As a result, we can then compute the optimal signaling scheme by restricting LP \eqref{lp:optPub}  to only these $\mathcal{O}(n^d)$ public signals.  However, there are two technical challenges for this approach. First, besides these labels, many other public signals actually can also arise in general. For instance, in Example \ref{ex:degeneracy}, we only have two cells but all the $2^n$ signals are valid due to (exponentially) many possible ways of tie breaking. This will be problematic when we search for the optimal signaling scheme. Fortunately, we prove in Lemma \ref{lem:correspondence} that under mild non-degeneracy assumption of receiver utilities, all the public signals that can possibly arise indeed correspond to the labels of all the $\mathcal{O}(n^d)$ cells.  Second, we need to know exactly what are these possible public signals in order to formulate the linear program. This is addressed by Lemma \ref{lem:cut-space} in which we design an algorithm to identify all these public signals, i.e., all the cells generated by $n$ hyperplanes. At a high level, the algorithm iteratively adds each hyperplane and identify all the newly generated cells at each iteration.  
\begin{lemma}\label{lem:correspondence}
	Suppose receiver payoffs are non-degenerate. Then all public signals that can possibly arise  are precisely all the labels of cells generated by the $n$ hyperplanes cutting $\RR^d$.  
\end{lemma}

\begin{lemma}\label{lem:cut-space}[Cutting $\RR^d$ with $n$ Hyperplanes]
	Any $n$ hyperplanes divide $\RR^d$ into $\O(n^d)$ cells. Moreover, all these cells can be identified (represented as intersections of $n$ half-spaces) in $\poly(n^d)$ time. 
\end{lemma}  


\begin{remark}
Theorem \ref{thm:const-poly} can be generalized to the setting where each receiver has $k$ actions instead of only two.  In this case, each receiver corresponds to $k(k-1)/2$ hyperplanes since we need to compare each pair of actions. 
We can similarly characterize all the possible public signals by examining the cells obtained from dividing $\Delta_d$ by $n \times \frac{k(k-1)}{2}$ hyperplanes, and obtain a $\poly(n^dk^{2d})$ time algorithm. 
\end{remark}


The key property that enables efficient computation of public persuasion --- even when directly optimizing the sender's objective function is completely intractable ---  is  its ``\emph{limited expressiveness}'' due to the restriction to public signals. That is, when  $d$ is small, public schemes can only induce a small set of outcomes (i.e., profiles of receiver actions). This is fundamentally due to the fact that every receiver derives the same posterior belief from a public signal.  As a consequence,  the sender's utility only depends on her objective function at these limited number of  outcomes. This is why the structure of the sender's objective functions did not play any role in our algorithm for public persuasion. In contrast,  the  computational complexity of private persuasion --- even when there are only two states of nature ---  is  governed by the complexity of  optimizing the sender's objective, as shown in \cite{Dughmi2017algorithmic}.

\section{Public Persuasion with Relaxed Persuasiveness}\label{sec:bicriteria}
We now consider public persuasion with relaxed receiver incentives.  That is, instead of insisting on exactly persuasive schemes, we allow the scheme to be $\epsilon$-persuasive for some small $\epsilon$ as formally defined in Inequalities \eqref{eq:eps-persuasive}. 
Here, our result concerns  the setting where the sender's objective function  is not affected by the state of nature, i.e., $f_{\theta} = f$ for all $\theta$.  This important special case has been the focus of many previous works including the original model by Arieli and Babichenko \cite{Arieli2016}, and is realistic in applications including, but not limited to, voting \cite{Alonso14}, marketing \cite{babichenko2016,Ozan19} and customer queuing \cite{Lingenbrink2019}.        


We design a \emph{tight} approximate algorithm for computing an $\epsilon$-persuasive public signaling scheme for \emph{submodular} sender objectives, a natural class of set functions which are widely used to model agents' utilities.  The crux in the analysis of this algorithm  is to prove an intriguing ``noise stability'' property of submodular functions, which strengthens the key result of Cheraghchi \emph{et al.} \cite{CKKL12}. In particular,  \cite{CKKL12} proved the noise stability of submodular functions under certain noise model motivated by applications of learning submodular functions and differential privacy. Our analysis proves exactly the same noise stability guarantee as in \cite{CKKL12} but under a strictly stronger (i.e., more adversarial) noise model. We provide a more detailed comparison  in Section \ref{sec:bicriteria:discuss}. It is an interesting open question  to see whether this strengthened noise stability property can lead to stronger guarantees for the applications mentioned in \cite{CKKL12}.



We first state the performances of our algorithm for public persuasion. 

\begin{theorem}\label{thm:sub}
Consider public persuasion with $f_{\theta} = f$ for any $\theta$. If $f$  is non-negative and \emph{monotone submodular},  then  there is a $\poly(n, |\Theta|^{1/\epsilon})$ time algorithm that outputs a $(1- \epsilon)$-approximate and $\epsilon$-persuasive public signaling scheme for any $\epsilon \in (0,1)$. If $f$  is \emph{non-monotone submodular}, then there is a  $\poly(n, |\Theta|^{1/\epsilon})$ time algorithm that outputs a $\frac{1}{2}(1- 2\epsilon)$-approximate and $\epsilon$-persuasive public scheme for any $\epsilon \in (0,1)$.
\end{theorem}

\begin{remark} The multiplicative approximation guarantee of the sender utility in Theorem \ref{thm:sub} is a consequence of the conventional multiplicative guarantees for submodular maximization. As we shall see later, the approximation guarantee of the sender utility is closely  related to  the problem of directly maximizing the function $f$. Indeed, any additive approximation algorithm for maximizing $f$ can also be used to compute a public signaling scheme with additive guarantee for the sender utility. 
\end{remark}

The next two results illustrate: (1) The approximation guarantees in Theorem \ref{thm:sub} are essentially the best possible; (2) Similar results are not possible for private persuasion, showing an interesting contrast between public and private persuasion.  
\begin{proposition}[Tightness of Theorem \ref{thm:sub}]\label{prop:sub:PTAS:best} 
For any constant $c > 0$, there is no $\poly(n, |\Theta|)$ time algorithm that computes a $c$-approximate and $(\frac{1}{4n})$-persuasive  public scheme even when $f_{\theta} = f = |S|$, unless P = NP.    
\end{proposition}

\begin{proposition}[The Contrast to Private Persuasion]\label{prop:private-PTAS}
Suppose $f_{\theta} = f$ for any $\theta \in \Theta$ and $f$ is monotone submodular. Unless P=NP, there is no $\poly(n, |\Theta|^{1/\epsilon})$ time algorithm that computes an $(1 - \epsilon)$-approximate and $\epsilon$-persuasive private signaling scheme even when there are only \emph{two} states of nature.  
\end{proposition}

Proposition \ref{prop:sub:PTAS:best} shows that there is no efficient algorithm that can achieve a constant approximation to the sender utility under the strengthened $\epsilon$-persuasiveness with $\epsilon = \poly{ (\frac{1}{n}) }$ instead of a constant $\epsilon$.   This proposition strengthened a hardness result in \cite{Dughmi2017algorithmic} by employing a tighter analysis for their reduction. The proof of Proposition \ref{prop:private-PTAS} leverages the APX-hardness of approximating the concave closure of monotone submodular functions due to \cite{babichenko2016} and ``converts'' it (through reductions)  to the NP-hardness of obtaining bi-criteria approximations.  The detailed proofs are  deferred to Appendix \ref{sec:bicriteria:PTAS:best} and \ref{sec:bicriteria:private:hard}, respectively.  

The rest of this section is devoted to the proof of Theorem \ref{thm:sub}. The starting point of our proof is a general framework developed in Cheng et al. \cite{mixture_selection}, which is useful for  designing approximately optimal signaling schemes.  However, our setting is  more challenging because the sender's objective  in our problem is a submodular function, which is APX-hard to maximize directly even when  persuasion or incentive is \emph{not} present, whereas in all the settings considered in \cite{mixture_selection}, directly maximizing their objective functions is an easy task. To prove our result, we first adapt and generalize the framework of \cite{mixture_selection} to incorporate certain notion of ``approximability'' of the objective function, defined as follows. Note that in this definition, we used the alternative notion of $\mathbf{s} \in \{ 0,1 \}^n$ to denote a subset, which appears more convenient. 

\begin{definition}[$\alpha$-Approximability]\label{def:maximizability} 
	A set function $f: \{ 0 , 1\}^n \to \RR_+$ is $\alpha$-approximable for some $\alpha \in [0,1]$ if there is a polynomial time algorithm such that for any $T \subseteq [n]$ and any given variable values $\mathbf{s}^0_{-T} \in \{ 0,1 \}^{-T}$ for variables in set $-T = [n] \setminus T$,  the algorithm finds a $\mathbf{s}^*_{T} $ such that 
	\begin{equation}\label{eq:alpha-approx}
	f(\mathbf{s}^*_{T}, \mathbf{s}^0_{-T})  \geq \alpha \cdot \max_{  \mathbf{s}_{T} \in \{ 0,1 \}^{T} } f(\mathbf{s}_{T}, \mathbf{s}^0_{-T}). 
	\end{equation}
	We will use such an algorithm frequently and will call it an $\alpha$-\emph{approximate subroutine} for $f$.  
\end{definition}

The request of $\alpha$-approximability is natural since otherwise it is even intractable to evaluate the sender's objective value by finding the optimal tie breaking for receivers in set $T$, which however is crucial for designing optimal signaling schemes. Observe that any monotone function is $1$-approximable since  $\mathbf{s}_T = \mathbf{1}$ is optimal for $\max_{  \mathbf{s}_{T} \in \{ 0,1 \}^{T} } f(\mathbf{s}_{T}, \mathbf{s}^0_{-T})$. Non-monotone submodular functions are $1/2$-approximable since unconstrained submodular maximization  admits a polynomial-time deterministic $1/2$-approximation \cite{buchbinder2018deterministic}.

Another key notion we need is the \emph{noise stability} of the sender objective. As mentioned previously, noise stability of set functions is not a new concept. Our definition here is a \emph{strictly stronger} version of the definition in \cite{CKKL12,Feldman2017tight}, and can be  viewed as a discrete variant of the noise stability defined in \cite{mixture_selection} for \emph{continuous} functions. We start by defining what ``noise'' mean in discrete cases.     

\begin{definition}[$\epsilon$-Noisy Distribution]\label{def:noisy-dist}
Let $S \subseteq [n] $ be any subset and $\mathbf{p}$ be a distribution over  $2^{[n]}$. For any
$\epsilon \in (0,1)$, we say $\mathbf{p}$ is an $\epsilon$\emph{-noisy distribution around $S$} if the following holds simultaneously: (1) for any $i \in S$: $\Pr_{T \sim \mathbf{p}}[i\in T] \geq 1- \epsilon$; (2) for any $i \in [n]\setminus S$: $\Pr_{T \sim \mathbf{p}}[i\in T] \leq \epsilon$.  
\end{definition}

In other words, $\mathbf{p}$ is an $\epsilon$-noisy distribution around $S \subseteq [n]$ if the sampled subset $T \sim \mathbf{p}$  ``almost'' equals  $S$, except that $T$ may exclude any $i \in S$ and include any $i \not \in S$, each with marginal probability at most $\epsilon$. Therefore, we can view $\Ex_{T \sim \mathbf{p}}f(T) $ as a \emph{noisy evaluation} of function $f$ at set $S$, during which $S$ is slightly perturbed.  Our definition of \emph{noise stability} lower bounds how much the expected function value would \emph{decrease} under   arbitrary noisy evaluation of $f$ at any subset $S$.

\begin{definition}[Noise Stability]\label{def:noise-stability}
A set function $f: 2^{[n]} \to \RR_+$ is  $\beta$-\emph{noise-stable}, or $\beta$-\emph{stable} for short, if for any $S \subseteq [n]$, any $\epsilon \in (0,1)$, and any $\epsilon$-noisy distribution $\mathbf{p}$ around $S$, we have:
\begin{equation}\label{eq:def-stable}
\Ex_{T \sim \mathbf{p}} f(T) \geq (1 - \beta \epsilon)f(S). 
\end{equation} 
	\end{definition}


It turns out that the approximability and noise stability of $f$ implies an efficient bi-criteria approximation for public persuasion, as stated in the following proposition whose proof is deferred to Appendix \ref{sec:bicriteria:app-framework}.  
 
\begin{proposition}\label{thm:bicriteria-approx}
Suppose $f_{\theta} = f$ for all $\theta$	 and $f$ is $\alpha$-approximable and $\beta$-stable,  then there is a $\poly(n, |\Theta|^{1/\epsilon})$ time algorithm that outputs an  $\alpha(1 - \beta\epsilon)$-approximate and $\epsilon$-persuasive public signaling scheme.  
\end{proposition}
Note that  Proposition \ref{thm:bicriteria-approx}   also holds  if both the $\alpha$-approximability and $\beta$-stability are defined in the additive manner. That is, if  $
f(\mathbf{s}^*_{T}, \mathbf{s}^0_{-T})  \geq  \max_{  \mathbf{s}_{T} \in \{ 0,1 \}^{T} } f(\mathbf{s}_{T}, \mathbf{s}^0_{-T}) - \alpha$ 
and $
\Ex_{T \sim \mathbf{p}} f(T) \geq f(S) - \beta \epsilon$, then 
there is a $\poly(n, |\Theta|^{1/\epsilon})$ time algorithm that computes an $(\alpha + \beta \epsilon)$-optimal (in additives sense)  and $\epsilon$-persuasive public scheme.

As mentioned previously, $\alpha$-approximability is easily satisfied by submodular functions, monotone ($1$-approximable) or non-monotone ($1/2$-approximable). So central to the proof of Theorem \ref{thm:sub} is to prove the $\beta$-noise-stability of submodular functions for small $\beta$, which is also the main technical novelty of our proof.   
Since this property of submodularity is also interesting in its own, we state it as a theorem. Note that Theorem \ref{thm:submodular-stable} and Proposition  \ref{thm:bicriteria-approx} yield a proof to Theorem \ref{thm:sub}. 

\begin{theorem}\label{thm:submodular-stable}
	Any submodular function $f: 2^{[n]} \to\RR_+$ is $2$-noise-stable. If $f$ is also monotonically non-decreasing, then $f$ is $1$-noise-stable. Moreover, the bound of $1$-noise-stability for monotone submodular functions is \emph{tight} even for non-negative linear functions.
\end{theorem}

\begin{proof}
	We will first consider general submodular functions without the monotonicity assumption. At the end of the proof, we illustrate how monotonicity allows us to tighten the analysis to achieve a better bound.  
	
	According to Definition \ref{def:noise-stability}, we need to prove that for any submodular function $f: 2^{[n]} \to \RR_+$, any $S \subseteq [n]$ and any $\epsilon$-noisy distribution $\mathbf{p}$ around $S$, we have $\Ex_{T \sim \mathbf{p}} f(T) \geq (1 - 2\epsilon)f(S)$. Equivalently, we can rephrase this requirement using the language of optimization, as follows:   among all possible $\epsilon$-noisy distributions around $S$, the minimum value of $\Ex_{T \sim \mathbf{p}} f(T)$ is at least $(1 - 2\epsilon)f(S)$. 
	
	Key to the proof is to formulate the noise stability requirement  as a problem of lower bounding the optimal objective of certain (exponentially large) linear program. We start by characterizing the set of all $\epsilon$-noisy distributions around $S$. Let $ \mathbf{p} = \{ p(T) \}_{T \subseteq [n]}$ denote a distribution over $2^{[n]}$. By Definition \ref{def:noisy-dist} $\mathbf{p}$  is an $\epsilon$-noisy distribution around $S$ if and only if the following constraints hold: (1) $\sum_{T: i \in T} p(T) \leq \epsilon$ for all $i \not \in S$; (2) $\sum_{T: i \in T} p(T) \geq 1 - \epsilon$ for all $i \in S$. Crucially, these are all linear constraints on $\mathbf{p} $. As a result, the minimum value of $\Ex_{T \sim \mathbf{p}} f(T)$ among all possible $\epsilon$-noisy distributions around $S$ corresponds precisely to the optimal objective value of the following  linear problem with variables $\{ p(T) \}_{T \subseteq [n]}$. The noise stability requirement can be equivalently stated as proving the optimal objective of LP \eqref{lp:worst-noise} is at least $(1-2\epsilon)f(S) $.  
	\begin{lp}\label{lp:worst-noise}
		\mini{\sum_{T \subseteq [n]} p(T) f(T)}
		\st
		\qcon{ \sum_{T: i \in T} p(T) \leq \epsilon }{i \notin S}
		\qcon{ \sum_{T: i \in T} p(T)\geq 1- \epsilon }{i \in S}
		\con{ \sum_{T \subseteq [n]} p(T) =1 }
		\qcon{p(T) \geq 0}{T \subseteq [n]}
	\end{lp}
	
	
	Next, we argue that  there always exists an optimal solution $\mathbf{p}^* =  \{ p^*(T)\}_{T \subseteq [n]}$ to LP \eqref{lp:worst-noise} which is a \emph{chain distribution}. That is, for any $T, T'$ in the support of $\mathbf{p}^*$, either $T \subset T'$ or $T' \subset T$. Let $\mathbf{p}^* $ be any optimal solution and $x_i^* = \sum_{T: i \in T} p^*(T)$ be the marginal probability that $i$ belongs to the random set $T \sim \mathbf{p}^*$. Then $\mathbf{p}^* $ can be viewed as the distribution over $2^{[n]}$  that minimizes the expected value of $f(T)$, i.e., $\sum_{T \subseteq [n]} p(T) f(T)$, among all distributions with marginal $\mathbf{x}^*$. The optimal objective is precisely the value of the \emph{Lov\'{a}sz extension} \cite{lovasz1983submodular} of a subsmodular function at point $\mathbf{x}^*$. It is widely known that the minimum objective can be achieved by a chain distribution (see, e.g., \cite{lovasz1983submodular,dughmi2009submodular,Arieli2016}).

	With the aforementioned property of $\mathbf{p}^*$, we are ready to lower bound the optimal objective of LP \eqref{lp:worst-noise}.  Let $\mathbf{p}^*$  be any optimal solution  that is a chain distribution and $x_i^* = \sum_{T: i \in T} p^*(T)$ be the corresponding marginal probability of $i$ belonging to the random set $T \sim \mathbf{p}^*$. Without loss of generality, assume that the elements are sorted in \emph{descending} order in term of $\{ x_i^* \}_{i \in [n]}$. That is, $x_{i}^* \geq x_{i+1}^*$ for all $i = 1,\cdots, n-1$. For convenience, let $x^*_{0} = 1$ and $x^*_{n+1} = 0$, which will be useful as notations.  
	
	With the above representation and the fact that $\mathbf{p}^*$ is a chain distribution, we know that $\mathbf{p}^*$   must be the following distribution:  $p^*(T_i) = x^*_{i} -x^*_{i+1}  (\geq 0)$, for all $i = 0,1,\cdots,n$,  where $T_i = \{ 1,...,i \}$ and $T_0$ is the empty set $\emptyset$.  We only consider $\epsilon < 1/2$ since otherwise $\Ex_{T \sim \mathbf{p}} f(T) \geq (1- 2 \epsilon)f(S)$ trivially holds. 
	The constraints in LP \eqref{lp:worst-noise} imply that for any $ i \in S$, we must have $x^*_i \geq 1-\epsilon >1/2$ and for any $i \not \in S$ we have $x^*_i \leq \epsilon <1/2$. Since the elements are  sorted in descending order of $x_i^*$,  the elements in $S$ must be exactly the first $|S| = k$ elements in the  sequence $1,\cdots, n$. In other words,  $S = \{ 1,...,k \} = T_{k} $. Moreover, $p^*(T_{k}) = x^*_{k} - x^*_{k+1} \geq 1-\epsilon - \epsilon $ since $x^*_{k} \geq 1-\epsilon $ and $ x^*_{k+1} \leq \epsilon$. As a result, we can lower bound the optimal objective value of LP \eqref{lp:worst-noise}, as follows:
	\begin{eqnarray*}
		\sum_{T \subseteq[n]} p^*(T) f(T) &=&  \sum_{i = 0}^n p^*(T_i) f(T_i) \\
		&=& p^*(T_{k}) f(T_{k}) +   \sum_{i \not = k } p^*(T_i) f(T_i) \\
		& \geq & (1-2\epsilon) f(S).
	\end{eqnarray*}
	Since $S$ and $\epsilon$ are chosen arbitrarily, this shows that any submodular function $f: 2^{[n]} \to \RR_+$  is  $2$-noise-stable.

	The analysis for monotone submodular functions is similar, but utilizes an additional property of minimizing a monotone non-decreasing  function. That is,  when $f$ is monotone, LP \eqref{lp:worst-noise} always admits an optimal solution such that $x^*_i = \sum_{T: i \in T} p^*(T)= 0$  for any $ i \not \in S$ due to monotonicity. Therefore, in the last step of the above analysis, we have $p^*(S) = x^*_k - x^*_{k+1} \geq 1 - \epsilon$ since  $x^*_{k+1} = 0$. This implies  $\sum_{T \subseteq[n]} p^*(T) f(T) \geq (1-\epsilon)f(S)$ and thus $f$ is $1$-stable. It is straightforward to verify that the bound of $1$-noise-stability is tight for, e.g., $f(S) = |S|$. 
	
\end{proof}

\subsection{Additional Discussions on Noise Stability and Its Connection to Previous Literature}\label{sec:bicriteria:discuss}

Here we provide some additional discussions on  noise stability  and its connection to previous literature. It is not difficult to see that any set function is $n$-noise-stable;  Moreover, supermodular functions are \emph{not} $\beta$-stable for any $\beta < n$ (see Appendix \ref{append2:0} for a proof).  Previously, only a few simple functions of concrete forms, e.g.,  $f(S) = |S|$ and $f(\mathbf{x}) = \maxs_{i \in [n]} x_i$ ($\maxs$ outputs the second largest value among $x_i$'s), are shown to  have small noise stability  \cite{mixture_selection}. To our knowledge, submodular functions are the first non-trivial yet rich class of  functions which are shown to all have small noise stability (under our notion).


Our definition of noise stability \emph{strictly} strengthens a previous definition by Cheraghchi et al. \cite{CKKL12}. The key difference between our definition and that of \cite{CKKL12} lies at the type of noise that is allowed. They consider the following noisy distribution, rephrased in our terminology for easier comparisons.  
	\begin{definition}[Independent-Noise Distribution, adapted from \cite{CKKL12}]\label{def:noisy-dist-JKKL}
		Let $S \subseteq [n] $ be any subset and $\mathbf{p}^I$ be any \emph{independent distribution} over $2^{[n]}$ (i.e., $S' \sim \mathbf{p}^I$ includes each $i \in [n]$ independently with probability $p^I_i$).  We say distribution $\mathbf{p}$ over  $2^{[n]}$ is  a $(\rho, \mathbf{p}^I)$\emph{-independent-noise distribution around} $S$  if its set $T \sim \mathbf{p}$ is generated via the following process: for any $i \in [n]$, with probability $\rho$  include $i$ to $T$ if and only if $ i \in S$, and with probability $1 - \rho$ include $i$ to $T$ independently with probability $\mathbf{p}^I_i$. 
		\end{definition} 	
		
That is, an independent-noise distribution around $S$ simply adds independent noise to $S$ with probability $\rho$. Observe that any $(\rho, \mathbf{p}^I)$-independent-noise distribution around $S$ is an $\epsilon$-noisy distribution in the sense of our Definition \ref{def:noise-stability} for $\epsilon = (1-\rho)(1  - p_{\min})$ where $p_{\min} = \min \{ p_i^{I}, 1 - p_i^{I} \}_{i \in [n]}$ is the minimum marginal probability of $\mathbf{p}^I$ --- in fact this is true even when $\mathbf{p}^I$ is an \emph{arbitrary} distribution, not necessarily independent. This is because each $i \in [n]$ is perturbed with probability at most $(1 - \rho)(1-p_{\min})$ in any $(\rho, \mathbf{p}^I)$-independent-noise distribution. Therefore, our notion of noise stability is \emph{strictly stronger}.    

		\begin{theorem}[Main Result of \cite{CKKL12}]\label{thm:CKKL}
			For any submodular function $f: 2^{[n]} \to \RR_+$ and any $(\rho, \mathbf{p}^I)$-independent-noise distribution $\mathbf{p}$, we have 
			\begin{equation*}
			\Ex_{T \sim \mathbf{p}} f(T) \geq (2 \rho - 1 + 2 p_{\min} (1-\rho))f(S).
			\end{equation*}
			\end{theorem}
Now we can see that Theorem \ref{thm:CKKL} follows from Theorem \ref{thm:submodular-stable} by observing $(2 \rho - 1 + 2 p_{\min} (1-\rho)) = 1 - 2 \epsilon$ for $\epsilon = (1 - \rho)(1-p_{\min})$. Thus, Theorem \ref{thm:submodular-stable} shows that submodular functions are noise-stable in a stronger sense  than shown previously but with exactly the same guarantee. We remark that the proof in  \cite{CKKL12} relies crucially on their independent noise assumption and does not appear applicable to our (stronger) notion of noise stability.   Our proof takes a completely different route and is more involved.  


\section{Public Persuasion under Relaxed Equilibrium Conditions}\label{sec:relax}
In this section, we consider public persuasion under relaxed equilibrium conditions. Classically in persuasion,  persuasiveness constraints require that  any recommended action to each receiver must indeed be his best response and thus it is in each receiver's best interest to  follow a direct scheme's recommendations. Here, we relax this constraint and assume that each receiver would follow the scheme's recommendations so long as the overall expected utility from the scheme is at least his best utility when acting based on his prior belief. This is motivated by applications where the persuasion scheme can be implemented as a software (e.g., recommendation systems like Google Maps, Yelp, etc.). Receivers have the options of either adopting the software (i.e., following the scheme's recommendations) or abandoning it (i.e., behaving according to their prior beliefs). As formally described in Definition \ref{def:cce}, we call such a scheme \emph{cce-persuasive} due to its relation to coarse correlated equilibrium.  

Our main result of this section  shows that the complexity of computing the optimal cce-persuasive public scheme is governed by the complexity of  directly maximizing the sender's objective function minus an arbitrary linear function. As a corollary, it implies that computing the optimal cce-persuasive public scheme is tractable for supermodular or anonymous sender's objective, and NP-hard for submodular objectives.  To describe our result, let $\F$ denote any collection of set functions. As a mild assumption, we assume $\F$ always contain the \emph{trivial function} $f \equiv 0$, i.e., $f$ always equals $0$ for any input.  Let $\I(\F)$ denote the class of all persuasion instances  in which the sender's  objective function $f_{\theta} \in \F$ for all $\theta \in \Theta$ but receiver payoffs can be arbitrary. We prove the following result.  

\begin{theorem}\label{thm: pub-max-equiv} 
	Let $\F$ be any collection of set functions that includes the trivial function $f\equiv 0$.  There is a polynomial time algorithm that computes the optimal cce-persuasive public scheme for any instance in $\I(\F)$ \emph{if and only if} there is a polynomial time algorithm for maximizing $f(S)-\sum_{i \in S} w_i $ subject to  no constraints for any $f \in \F$  (with value oracle access) and any set of weights $w_i \in  \RR$.  
\end{theorem}
\noindent {\bf Proof Sketch. } We provide a sketch here and defer the full proof to Appendix \ref{sec:relax:app}. The proof is divided into two parts.   First, we show that computing the optimal cce-persuasive public scheme reduces in polynomial time to the combinatorial problem of maximizing set function $\bar{f}(S) = f(S) - \sum_{i \in S} w_i$ for any $f \in \F$ and weights $\{ w_i \}_{i=1}^n$, subject to no constraints. This  can be shown by formulating the persuasion problem as an exponentially large linear program and then examining its dual program. It turns out that the set function optimization problem can be used to construct a separation oracle for the dual.

More involved is the reverse direction. That is, given any efficient algorithm for solving any persuasion instance from $\I(\F)$, we need to design an efficient algorithm for the set function maximization problem.  Given any $f \in \F $ and a set of weights $\{ \bar{w}_i \}_{ i \in [n]}$, we look for an efficient algorithm to maximize $\bar{f}(S) = f(S) - \sum_{i \in S} \bar{w}_i$ over $S \subseteq [n]$.  Let $S_+ = \{ i: \bar{w}_i \geq 0 \}$ and $ S_- = \{ i: \bar{w}_i < 0 \}$ denote the set of indexes of non-negative and negative $\bar{w}_i$'s, respectively.  The first key step is to reduce the maximization of $\bar{f}(S)$ to solving the following carefully constructed linear program with  variables $\{ w_i\}_{ i \in [n]}$ and $v$. Our reduction utilizes   the equivalence between separation and linear optimization \cite{GLSbook}, and then binary search for the smallest feasible $v$ which must equal $\max_{S} [f(S) - \sum_{i \in S} w_i]$. 

\begin{lp}\label{lp:WM_Dual_0}
	\mini{\sum_{ i \in  [n]} \beta_i \cdot w_i  + \alpha \cdot v} 
	\st
	\qcon{ v + \sum_{i \in S} w_i  \geq f(S) }{ S \subseteq [n]}  
	\qcon{w_i \geq 0}{i \in S_+}
	\qcon{w_i \leq 0 }{i \in S_-} 
\end{lp}

Then we want to show that solving the persuasion problem can help us to solve the \emph{dual program} of LP \eqref{lp:WM_Dual_0} (note: LP \eqref{lp:WM_Dual_0}  itself is a minimization problem, which is not appropriate). Unfortunately, this is \emph{not} true when linear coefficient $\alpha \leq 0$. However, we can prove that LP \eqref{lp:WM_Dual_0} can be directly (and efficiently) solved when $\alpha \leq 0$ by analyzing the structure of the program. Therefore, we only need to worry about the case with $\alpha > 0$, for which we manged to construct a persuasion instance such that its optimal cce-persuasive public scheme can be used to derive an optimal solution to the dual of LP \eqref{lp:WM_Dual_0}. 

 We remark that  our proof of Theorem \ref{thm: pub-max-equiv} builds on and generalizes a duality-based analysis employed in \cite{Dughmi2017algorithmic},  which proved a similar equivalence theorem but for optimal private persuasion. However, the analysis in \cite{Dughmi2017algorithmic}  relies crucially on the assumption that the sender's objective functions are \emph{monotone non-decreasing}. Our result  gets rid of the monotonicity assumption completely and thus applies to a broader class of objective functions. This, however, also necessitates a more involved proof. Our argument relies on a more carefully constrained linear program  as well as the  analysis of its properties  in order to avoid the dependence on the objective function's monotonicity. In the full proof, we highlight the key differences from the proof of \cite{Dughmi2017algorithmic}.

\subsection*{Acknowledgement} The author would like to thank anonymous reviewers for helpful comments.

\bibliographystyle{ACM-Reference-Format}
\bibliography{refer}
\newpage 
\appendix
\section{Omissions from Section \ref{sec:constant} }

\subsection{Proof of Theorem \ref{thm:const-poly}}\label{sec:const:app}

	Recall that in (direct) public schemes, each signal can be viewed as a subset  $S \subseteq [n]$ of receivers who are incentivized to take action $1$ while the remaining receivers are incentivized to take action $0$. For convenience, we sometimes also denote a subset $S$ as a binary vector $\mathbf{s} \in \{ 0, 1\}^n$, satisfying $s_i = 1$ if and only if $i \in S$.  The main difficulty in designing optimal public scheme is that there are possibly $2^n$ different public signals and  any efficient algorithm simply cannot search over all these exponentially many signals. Our key insight is that for a  persuasion instance with small $|\Theta| = d$, most of these public signals actually will never arise in \emph{any} public signaling scheme, assuming non-degenerate receiver payoffs. In fact, for any non-degenerate persuasion instance,  there are only $O(n^d)$ signals that can possibly arise in public schemes and, moreover, these public signals can be efficiently identified in $\poly(n^d)$ time. This turns out to be a consequence of  the problem of dividing $\RR^d$ by hyperplanes. 
	
	Specifically, each public signal induces a posterior distribution $p \in \Delta_d$ over the states in  $\Theta$.  We use $p_{\theta}$ to denote the probability of $\theta \in \Theta$. Note that, any posterior distribution that induces receiver $i$ to take action $1$ [resp. $0$] must satisfy $\langle u_i, p \rangle = \sum_{\theta \in \Theta} u_i(\theta) p_{\theta}  \geq 0$ [resp.  $\sum_{\theta \in \Theta} u_i(\theta) p_{\theta}  \leq 0$]. In other words, the hyperplane $\sum_{\theta \in \Theta} u_i(\theta)p_{\theta}  =  0$ cut the simplex $\Delta_d$ of posterior distributions into at most two cells, each corresponding to receiver $i$'s best response of action $1$ and $0$, respectively.  With $n$ receivers,  the simplex will be cut by $n$ hyperplanes and each cell --- the intersection of $n$ half-spaces restricted to the simplex $\Delta_d$--- is uniquely characterized by a $n$-dimensional \emph{binary vector} $\mathbf{s} \in \{ 0,1 \}^n$, in which the $i$'th entry $s_i$ indicates receiver $i$'s best response action (action $1$ for ``$\geq$'' and action $0$ for ``$\leq$'').   We  call vector $\mathbf{s}$ the \emph{label} of that cell. Since  $\mathbf{s}$ also specifies each receiver's best response. We will also refer to it as ``response vector'' and will use the term ``label'' and ``response vector'' interchangeably.  Note that cells may intersect at some boundary, which will correspond to multiple labels due to ties.

	
%
	
	Each cell has a unique label, however the reverse may not be true due to a subtle issue. That is,  there may exist response vectors that can arise but do not correspond to any of the generated cells.  Example \ref{ex:degeneracy} provides such an instance --- there are only $2$ cells but $2^n$ possible public signals. It turns out that this difficulty is due to the degeneracy of receiver payoffs. Assuming non-degenerate receiver payoffs, we can characterize all the response vectors that can  possibly arise and upper bound the total number of them, as described in the following two lemmas. For convenience, in Lemma \ref{lem:correspondence_app} and \ref{lem:cut-space_app}, we will describe the characterizations by \emph{relaxing the domain of $p$
		from $\Delta_{\Theta}$ to $\RR^{\Theta}$}. At the end of the proof, we show how to adapt the characterization for  $ p \in \Delta_{\Theta}$. 
	
	\begin{lemma}[Restating Lemma \ref{lem:correspondence}]\label{lem:correspondence_app}
		Suppose receiver payoffs are non-degenerate and $p $ takes values in $\RR^{\Theta}$.  Then all response vectors that can possibly arise  are precisely all the labels of cells generated by the $n$ hyperplanes in $\RR^{\Theta}$. 
	\end{lemma}
	\begin{proof}
			Let $\bar{\Sigma} \subseteq \{ 0,1 \}^n$ denote the set of all response vectors that can possibly arise, and $\mathcal{S}\subseteq \{ 0,1 \}^n$ denote the set of labels of all the cells generated by cutting $\RR^{\Theta}$ with $n$ hyperplanes: $ \sum_{\theta \in \Theta} u_i(\theta)p_{\theta}  =  0$, denoted as $l_i$,  for $  i = 1, \cdots, n$. We will prove $\bar{\Sigma} = \mathcal{S}$.

			Note that by definition we have $\mathcal{S} \subseteq \bar{\Sigma}$ since any $\mathbf{s} \in \S$ is the label of a cell, i.e., a possible response vector. We prove the reverse $\bar{\Sigma} \subseteq  \mathcal{S} $. In particular, we show that for any $p \in \RR^{\Theta}$, any possible response vector it induces is contained in $\mathcal{S}$. This follows a discussion about which subspace $p$ sits in. If $p \in \RR^{\Theta}$ is in the \emph{interior} of some cell, then the only response vector it can induce is the label of that region. 
			
			Let $|\Theta| = d$. More generally, if $p $ is in the \emph{interior} of some $(d-k)$-dimensional subspace $V$ generated as the intersection of our hyperplanes, we claim that $V$ must be the intersection of exactly $k$ hyperplanes and thus $p$ will only be on these $k$ hyperplanes. Since $V$ is a $(d-k)$-dimensional subspace, it is the intersection of at least $k$ hyperplanes. Without loss of generality, let $l_1,...,l_m$ be these hyperplanes for some $m \geq k$. We show that $m$ must equal $k$. By definition, there exists ``origin'' $x_0 \in \RR^{\Theta}$ and linearly independent directions $x_1,\cdots, x_{d-k} \in \RR^{\Theta}$ such that any $ p \in V$ can be expressed as $x_0 + \sum_{j=1}^{d-k} c_j \cdot x_{j}$ for some coefficients $\{ c_j\}_{j=1}^{d-k}$.  Since $V$ is the intersection of hyperplanes $l_1,\cdots, l_m$, we know that $\langle u_i , x_j \rangle = 0$ for any $i = 1, \cdots, m$ and $j = 1, \cdots, d-k$. Since $x_j$'s are linearly independent, this means the space spanned by vectors $u_1,\cdots, u_m$ have dimension at most $d - (d-k) = k$. By our assumption of non-degenerate receiver payoffs, we know that  $u_1,\cdots, u_m$  are linearly independent and thus the space they span has dimension at least $m$. As a result, we have $m \leq k$, implying $m = k$. 
			
			Therefore,  any $p $ in the \emph{interior} of some $(d-k)$-dimensional subspace  will be on exactly $k$ hyperplanes, which we denote as $l_1,\cdots, l_k$ for referral convenience. As a result, the number of response vectors it can induce is $2^k$ due to ties on the $k$ receivers corresponding to these $k$ hyperplanes. Note that $p$ is not on hyperplane $l_{k+1}, \cdots, l_n$, thus there exists a small ball in $\RR^{\Theta}$ centered at $p$ which does not intersect $l_{k+1}, \cdots, l_n$ neither. However, this ball will intersect half spaces $ \sum_{\theta \in \Theta} u_i(\theta)p_{\theta}  > 0$ and $ \sum_{\theta \in \Theta} u_i(\theta)p_{\theta}  < 0$  for $i = 1, \cdots, k$, and thus intersect with exactly $2^k$  cells. The labels of these cells are precisely the $2^k$ response vectors that $p$ can induce. 
			
			By varying $k = 0, \cdots, d$, we know that the response vectors induced by any point $p \in \RR^{\Theta}$  must be contained in $\S$. We thus have $\bar{\Sigma} \subseteq \S$, concluding the proof.  
	\end{proof}

	Next, we prove that $\RR^{\Theta}$ will be divided into $\O(n^d)$ cells by $n$ hyperplanes and moreover, each of these cells can be explicitly identified as the intersection of $n$ half-spaces. 
	\begin{lemma}\label{lem:cut-space_app}[Restating Lemma \ref{lem:cut-space}]
		Any $n$ hyperplanes divide $\RR^d$ into $\O(n^d)$ cells. Moreover, all these cells can be identified (represented as intersections of $n$ half-spaces) in $\poly(n^d)$ time. 
	\end{lemma}  
	\begin{proof}
	
 The study about how many cells the space of $\RR^d$ will be divided into by $n$ hyperplanes dates back to 1820s, first studied by J. Steiner  \cite{steiner1826einige}. Since then, this question and its many variants  attracted much attention in the mathematical literature (see, e.g., \cite{buck1943partition,alexanderson1981arrangements} and references therein).   
	The following fact  can be easily proved by establishing an induction relation.\footnote{An interesting side note: this problem was also portrayed in the film \emph{Let Us Teach Guessing} about the famous mathematician  George P{\'o}lya to illustrate induction.} 
	
	
	\begin{fact}\label{fact:NumRegions}
		$n$ hyperplanes divide the space of $\RR^d$ into  at most  $\sum_{i=0}^d \binom{n}{i} = \O (n^d)$ cells.
	\end{fact}
	
	To prove Theorem \ref{thm:const-poly}, we  need to know not only how many cells  there  are but also what are these cells, i.e., which linear inequalities generate each cell (i.e., a polyhedron).  
	Here, we describe an algorithm to compute, for each cell, the linear inequalities that generate the cell. Recall that each cell can be uniquely determined by a binary vector, which we call its \emph{label}.  
	
%
	Details are in Algorithm \ref{alg:labeling}. At a high level, our algorithm starts with the whole space $\RR^d$ and gradually adds each hyperplane. For each cell, we check whether the added hyperplane cuts the cell into two cells in which case  we create labels for newly generated cells (Step \ref{step:1} to \ref{step:2} )  or the cell is strictly on one side of the hyperplane in which case we augment the label of that cell by taking into account the newly added hyperplane (Step \ref{step:3} to \ref{step:4}). The correctness of this algorithm follows from its definition and the fact that feasibility of a cell determined by linear inequalities can be checked efficiently. 
	
	\end{proof}

	\begin{algorithm}[t]
		\SetAlgoNoLine
		\KwIn{$n$ hyperplanes in $\RR^d$: $\mathbf{a}_i \cdot \mathbf{x} = b_i$ for $i = 1,\cdots, n$}
		\KwOut{A list of the labels of all the cells generated by these hyperplanes.}
		\vspace{2mm}
		Initialization: $\texttt{LabelSet} = \{  [\, ] \}$, meaning starting with the only cell $\RR^d$ of label length $0$ \; 
		\For{$ i = 1, \cdots, n$}{
			\For{each cell $\mathbf{s} \in $ \texttt{LabelSet}}{
				\tcc{Extract linear inequalities, stored in \texttt{IneqSet},  that determine cell $\mathbf{s}$}
				\texttt{IneqSet}$=\{ \}$\;
				For $j \in [i-1]$, add inequality $\mathbf{a}_j \cdot \mathbf{x} \geq b_j$ to \texttt{IneqSet} if $s_j = 1$; otherwise, add $\mathbf{a}_j \cdot \mathbf{x} \leq b_j$ to \texttt{IneqSet}\;
				\tcc{Compute the new cells generated by cutting $\mathbf{s}$ with hyperplane $\mathbf{a}_i \cdot \mathbf{x} = b_i$}  
				\uIf{\label{step:1} both \texttt{IneqSet}$\cup \{ \mathbf{a}_i \cdot \mathbf{x} \geq b_i\}$ and \texttt{IneqSet}$\cup \{ \mathbf{a}_i \cdot \mathbf{x} \leq b_i\}$ are feasible}{
					Remove $\mathbf{s}$ from \texttt{LabelSet} \;
					Add  $[\mathbf{s};1]$ and $[\mathbf{s};0]$ to \texttt{LabelSet};  \label{step:2} 
				}
				\uElseIf{\label{step:3} only \texttt{IneqSet}$\cup \{ \mathbf{a}_i \cdot \mathbf{x} \geq b_i\}$ is feasible}{
					Substitute $\mathbf{s}$ in \texttt{LabelSet} by $[\mathbf{s};1]$;
				}
				\Else{ Substitute $\mathbf{s}$ in \texttt{LabelSet} by $[\mathbf{s};0]$; \label{step:4} } 
			}
		} 
		Return \texttt{LabelSet}. 
		\caption{Labeling Algorithm for Region Identification}
		\label{alg:labeling}
	\end{algorithm}
	Lemma \ref{lem:correspondence} and \ref{lem:cut-space} together characterize all the vectors of follower best responses, i.e., all the public signals,  that can possibly arise. The only discrepancy here is that we relaxed the domain of $p$ to $\RR^{\Theta}$. If we restrict $p \in \Delta_{\Theta}$, some of these public signals will be further eliminated. This can be done efficiently by examining these public signals one by one. We end up with all the possible public signals that can arise for $p \in \Delta_{\Theta}$, and there are $\O (n^d)$  of them.  
	
	Now let $\bar{\Sigma}$ denote the set of all these public signals that can possibly arise in the persuasion instance. We can then compute the optimal public signaling scheme simply by solving LP \eqref{lp:optPub} but restrict it to signal space $\bar{\Sigma}$ instead of the whole space $\Sigma = 2^{[n]}$.  This LP can be solved in $\poly(n^d)$ time.

\subsection{Additional Applications of the  Technique}\label{sec:constant:applictions}
In this subsection, we show another application of our technique to the problem of \emph{signaling in second price auctions} to maximize the auctioneer's revenue. This problem has been studied in several previous works, partially driven by its wide application in online advertising. The general model, first studied by \cite{Emek12}, considers an auctioneer (sender) facing $n$ bidders (receivers). The problem is generally described by a tuple  $ \langle \lambda, q, \{ V_{t} \}_{t \in [T]} \rangle $, where $V_t \in \RR^{n \times \Theta}, \forall t \in [T]$ and $V_t(i,\theta)$ is bidder $i$'s value at the state $\theta$. The auctioneer's uncertainty regarding bidder values is casted by valuation type $t \in T$. That is, the auctioneer only knows that value matrix is $V_t$ with probability $q_t$. The public persuasion problem of the auctioneer is to design a public signaling scheme to maximize her expected revenue. 

Emek et al. \cite{Emek12} prove that  this problem is NP-hard even when there are $n = 3$ bidders but admits a polynomial time when the number of value types $|T|$ is a constant. Invoking Lemma \ref{lem:cut-space}, we complete the picture of the fix parameter tractability of  this problem by showing that a polynomial time algorithm exits when $|\Theta|$ is a constant. Note that Cheng et al. \cite{mixture_selection} designs an additive PTAS for this problem, which however remains a PTAS even when $|\Theta|$ is a constant. 
\begin{proposition}\label{prop:auction-poly}
	The optimal public signaling scheme for the above second-price auction model can be computed in $\poly(n, |T|)$ time when $|\Theta| $ is a constant. 
\end{proposition} 

\begin{proof}
Similar to the proof of Theorem \ref{thm:const-poly}, our idea is also to argue that the number of ``outcomes'' that can possibly arise is polynomial in $n, |T|$. The key to our proof is a properly chosen definition of ``outcome'' in this setting.  Recall that in second-price auctions, the auctioneer will rank the bidders' values and any ranking can be equivalently viewed as a permutation $\pi$ where $\pi(i) \in [n]$ denotes the bidder whose value is the $i$'th largest. 
Given any posterior $p \in \Delta_{\Theta}$, we define an outcome $o = \{ \pi_t \}_{t \in T} $ of $p$ as a set of permutations  where $\pi_t$ denotes the ranking of bidders based on their expected values under the posterior distribution for value type $t$. Let $\Pi$ denote the set of all permutations over $[n]$. A naive counting argument reveals that there are at most $(n!)^{|T|}$ outcomes that can possibly arise.  However, invoking Lemma \ref{lem:cut-space}, we show that only polynomially many outcomes can arise when $|\Theta|$ is a constant, and moreover all these outcomes can be identified efficiently.  

Let $p \in \Delta_{\Theta}$ denote a generic posterior distribution.  We define the \emph{comparison hyperplane} for bidder $i,j$ under value type $t$, denoted as $l(i,j;t)$, as follows
$$l(i,j;t): \quad \sum_{\theta} V_t(i,\theta) p_{\theta} = \sum_{\theta} V_t(j,\theta) p_{\theta}.$$ 
Note that $l(i,j;t)$ is a hyperplane with vector variables $p$ and there are $\frac{n(n-1)}{2}|T|$ comparison hyperplanes. Now consider $\Delta_{\Theta}$ cut by these comparisons hyperplanes. By Lemma \ref{lem:cut-space}, they generate  $\O \bigg( (\frac{n(n-1)}{2}|T|)^{|\Theta|}\bigg)$ regions, which can be identified efficiently. Moreover, we argue that each region corresponds to a unique outcome $o$. In particular, the interior of any region is a set of posterior distributions. Depending on the side of $l(i,j;t)$  the region is on, we can tell whether bidder $i$ has higher or lower value than bidder $j$ at value type $t$. Aggregating this information across $i,j \in [n]$ and $t \in T$, we can extract an outcome $o = \{ \pi_t \}_{t \in T} $ for this region. Moreover, it is easy to see that different regions will have different outcomes, and let $O$ denote the set of all outcomes that can possibly arise. 

After identifying $O$ and viewing it as the set of needed signals, we can compute the optimal signaling scheme by the following linear program where variable $\varphi(\theta, o)$ is the probability of sending signal $o$.
\begin{lp}
	\maxi{  \sum_{o \in O}  \sum_{t \in T} \sum_{\theta}  \lambda_{\theta} \varphi(\theta, o) V_t(\pi^o_t(2), \theta)  }
	\st 
	\qcon{  \sum_{\theta}  \lambda_{\theta} \varphi(\theta, o) [ V_t(\pi^o_t(i), \theta) - V_t(\pi^o_t(i+1), \theta) ]  \geq  0    }{i \leq  n-1, o \in O, t \in T}
	\qcon{\sum_{o \in O}\varphi(\theta, o) = 1}{\theta \in \Theta}
	\qcon{\varphi(\theta, o) \geq 0}{\theta \in \Theta, o \in O}
\end{lp}
Proposition \ref{prop:auction-poly} then follows from the fact $|O| \leq \O \bigg( (\frac{n(n-1)}{2}|T|)^{|\Theta|}\bigg)$.
\end{proof}


\newpage 
\section{Omitted Proofs from Section \ref{sec:bicriteria}}\label{sec:bicriteria:app}

\subsection{Proof of Proposition \ref{prop:sub:PTAS:best} }\label{sec:bicriteria:PTAS:best}
Similar to the proof  in \cite{Dughmi2017algorithmic}, we also reduce from the same NP-hard problem described as follows but use a slightly different construction of the persuasion instance to accommodate the relaxation of persuasiveness. In particular, \cite{Khot2012} prove that for any positive integer $k$, any integer $q$ such that $q \geq  2^k+ 1$, and an arbitrarily small constant $\eps > 0$,  given an undirected graph $G$, it is NP-hard to distinguish between the following two cases:
	\begin{itemize}
		\item {\bf Case 1}: There is a $q$-colorable induced subgraph of $G$ containing a  $(1-\eps)$ fraction of all vertices, where each color class contains a $ \frac{1-\eps}{q}$ fraction of all vertices.
		\item{\bf Case 2}: Every independent set in $G$ contains less than a $\frac{1}{q^{k+1}}$ fraction of all vertices.
	\end{itemize} 
	Given graph $G$ with vertices $[n] = \set{1,\ldots,n}$ and edges $E$, we will construct a public persuasion instance so that any desired bi-criteria approximate signaling scheme can be used to distinguish these two cases. Let there be $n$ receivers, and let $\Theta= [n]$. In other words, both receivers and states of nature correspond to vertices of the graph. We fix the uniform prior distribution over states of nature --- i.e., the realized state of nature is a uniformly-drawn vertex in the graph. We define the receiver utilities as follows: $u_{i}(\theta) = \frac{1}{2} - \frac{1}{4n}$ if $i = \theta$;  $u_{i}(\theta) =-1 - \frac{1}{4n}$ if $(i,\theta) \in E$;  and $u_{i}(\theta) = -\frac{1}{2n}$ otherwise. We define the sender's utility function, with range $[0,1]$, to be $f_{\theta}(S) = f(S) =|S|$. 
	
	We claim that for any distribution $x \in \Delta_{\Theta}$, the set $S = \{ i \in [n]: \sum_{\theta \in \Theta} x_{\theta} u_i(\theta) \geq -\frac{1}{4n} \}$ is an independent set of $G$. In particular, for any two adjacent nodes $i,j$, if $x_i \geq x_j$, we have 
	\begin{eqnarray*}
\sum_{\theta} x_{\theta} u_j(\theta) &\leq& x_j u_j(j) + x_i u_j(i) + \sum_{\theta \not = i,j} x_{\theta} (-\frac{1}{2n}) \\
& = & x_j (\frac{1}{2} - \frac{1}{4n}) + x_i (  -1 - \frac{1}{4n} ) - (1-x_i-x_j) \frac{1}{2n} \\
& <  & - \frac{1}{4n}.
	\end{eqnarray*}


Thus, node $i \not \in S$. Therefore, at most one of any two adjacent nodes are in $S$. This shows that upon receiving any public signal with any posterior distribution $x$ over $\Theta$, the players who take action $1$ assuming $\frac{1}{4n}$-persuasiveness always form an independent set of $G$. Therefore, if the graph $G$ is from {\bf Case 2}, the sender's expected utility in any $\frac{1}{4n}$-persuasive public signaling scheme is at most $\frac{n}{q^{k+1}}$.
	

Now supposing that $G$ is from {\bf Case 1}, we fix the corresponding coloring of $(1-\epsilon)n$ vertices with colors $k=1,\ldots,q$, and we use this coloring to construct a public scheme achieving expected sender utility at least $\frac{(1-\epsilon)^2}{q}$. The scheme uses $q+1$ signals, and is as follows: if $\theta$ has color $k$ then deterministically send the signal $k$, and if $\theta$ is uncolored then deterministically send the signal $0$. Given signal $k > 0$, the posterior distribution on states of nature is the uniform distribution over the vertices with color $k$ --- an independent set $S_k$ of size $\frac{1-\epsilon}{q} n$. It is easy to verify that receivers $ i \in S_k$ prefer action $1$ to action $0$, since $\sum_{\theta \in S_k} \frac{1}{|S_k|} u_i(\theta) = \frac{1}{|S_k|}  (\frac{1}{2} - \frac{1}{4n} - \frac{|S_k| - 1}{2n}) = \frac{1}{|S_k|} ( \frac{2n - 1 - 2|S_k|+1}{4n} )  > 0$. Therefore, the sender's utility is $f(S_k) = |S_k| = \frac{n(1-\epsilon)}{q}$ whenever $k>0$. Since signal $0$ has probability  $\epsilon$, we conclude that the sender's expected utility is at least $\frac{n(1-\epsilon)^2}{q}$. It is now easy to see that there is no, e.g., $c$-approximate  $\frac{1}{4n}$-persuasive public signaling scheme for any constant $c$ since it can be used to distinguish {\bf Case 1} (utility at least $\frac{n(1-\epsilon)^2}{q}$) and {\bf Case 2} (utility at most $\frac{n}{q^{k+1}}$) for any constants $q, k, \epsilon$.

	

\subsection{Proof of Proposition \ref{prop:private-PTAS} }\label{sec:bicriteria:private:hard}

Previously, Babichenko and Barman \cite{babichenko2016} proved that  it  is NP-hard to obtain a $(1-\frac{1}{e})$-approximate, exactly persuasive, private signaling scheme for this setting. The main technical challenge of our proof is to ``transform'' this single-criteria approximation hardness result to a bi-criteria approximation hardness. To do so, we will actually instead prove the hardness of additive approximation of sender utility by assuming $f(S) \in [0,1]$. The hardness of additive approximation then implies the hardness of multiplicative approximation. 

We consider the following  persuasion instances. There are $n$ receivers  and two states of nature $\theta, \theta'$, each happening with equal probability $1/2$. At any state, the sender's objective  is the same monotone submodular function $f: 2^{[n]} \to [0,1]$. All the receivers have identical payoff structures. In particular, for any $i$, the utility of receiver $i$ is $a \in (0,1)$  in the state $\theta'$ for some constant $a$ to be chosen and $-1$ in the state $\theta$, i.e., $u_i(\theta') = a$ and $u_i(\theta) = -1$ for all $i$.

We start by formulating the problem of computing the optimal $\epsilon$\emph{-persuasive} private signaling scheme, which is central to our reduction.  Observe that any receiver $i$ strictly prefers action $1$ at the state of nature $\theta'$ since $u_i(\theta') = a$. Since $f(S)$ is monotone, the optimal private scheme can, without loss of generality, always recommend action $1$ to receiver $i$ at the state $\theta'$.  As a result, the  optimal $\epsilon$-persuasive private signaling scheme can be formulated as the following linear program, where the variables $\{ p(S) \}_{S \subseteq [n]}$ describe the signaling scheme for the state $\theta$ whereas the scheme always recommend action $1$ in state $\theta'$. That is, $p(S)$ is the probability of recommending action $1$ to receivers in set $S\subseteq [n]$ at the state of nature $\theta$. 

\begin{lp}\label{lp:private_reduction}
	\maxi{ \frac{1}{2}f([n]) + \frac{1}{2} \sum_{S\subseteq [n]} p(S)  f(S) }
	\st 
	\qcon{ \frac{1}{2}[ a \cdot 1 + (-1) \cdot \sum_{S: i \in S } p(S) ] \geq -\epsilon }{i = 1,...,n}
	\con{\sum_{S \subseteq [n]} p(S) = 1}
	\qcon{ p(S) \geq 0}{S \subseteq [n]}
\end{lp}
Note that LP \eqref{lp:private_reduction} slightly differs from the standard formulation for private persuasion in that it relaxed the persuasiveness constraint for action $1$ by $\epsilon$. Moreover, the (relaxed) persuasiveness constraint for action $0$, which is $\frac{1}{2}\cdot (-1) \cdot \sum_{S: i \not  \in S } p(S)  \leq \epsilon $, trivially holds and thus is omitted.   After algebraic simplifications, LP \eqref{lp:private_reduction} becomes the following equivalent linear  program. 

\begin{lp}\label{lp:concave-closure}
	\maxi{ \sum_{S\subseteq [n]} p(S)  f(S) }
	\st 
	\qcon{ \sum_{S: i \in S } p(S)  \leq a  + 2\epsilon }{i = 1,...,n}
	\con{\sum_{S \subseteq [n]} p(S) = 1}
	\qcon{ p(S) \geq 0}{S \subseteq [n]}
\end{lp}
Observe that LP \eqref{lp:concave-closure} is precisely the formulation for computing  the  \emph{concave closure} of the set function $f(S)$ at vector $(a + 2\epsilon, a+2\epsilon, \cdots, a+2\epsilon)$. The only discrepancy here is that the first set of constraints have ``$\leq$'' while not ``$=$'' as in a standard formulation for computing concave closure. However, this is without loss of generality since $f(S)$ is monotone non-decreasing and thus there always exists an optimal solution which makes the first set of constraints all tight.  For convenience, we use $f^+(x)$ to denote the concave closure of $f$ at the vector with all entries equaling $x \in [0,1]$, i.e., $(x, x, \cdots, x)$.   In our constructed instances, the optimal sender utility among all $\epsilon$-persuasive private signaling schemes  is $\frac{1}{2}[f([n]) + f^+(a+2\epsilon)]$. As a special case, $\frac{1}{2}[f([n]) + f^+(a)]$ is the optimal sender utility with exact persuasiveness constraints. The following lemma shows that it is NP-hard to approximate $f^+(x)$ additively.  

\begin{lemma}\cite{Khot2005,babichenko2016}\label{lem:submodular-PTAS-hard}
There exists a constant $c \in (0, 1) $ and constant $a \in (0, 1)$ such that it is NP-hard to approximate $f^+ (a)$  to within $c$ (additively) for any monotone submodular function $f$.\footnote{This lemma is  used in \cite{babichenko2016} and is a  slight generalization of results from \cite{Khot2005}. Both \cite{Khot2005,babichenko2016} used the multiplicative version of the lemma -- i.e., the concave closure cannot be approximated within a multiplicative factor of $(1-1/e)$. However, the reduction in \cite{Khot2005} can be easily adjusted to also prove the hardness for additive approximation by normalizing their constructed submodular functions to be within $[0,1]$. }
\end{lemma}

In the remainder of the reduction, we show that any $(1-\delta)$-approximate $\epsilon$-persuasive private persuasion in the constructed instances can be converted to an additively $(2\epsilon/a+2\delta)$-optimal algorithm  for computing $f^+(x)$. Invoking Lemma \ref{lem:submodular-PTAS-hard}, this implies the NP-hardness of designing  a  $(1-\epsilon)$-approximate $\epsilon$-persuasive private signaling scheme in $\poly(n^{1/\epsilon}, |\Theta|)$ time.

Let $\hat{\mathbf{p}} = \{ \hat{p}(S) \}_{S \subseteq [n]}$ be any $\epsilon$-persuasive and $(1 - \delta)$-approximate private signaling scheme for our constructed instance. By definition, $\hat{\mathbf{p}}$ is a feasible solution to LP \eqref{lp:private_reduction} and thus \eqref{lp:concave-closure}. We only need to argue that the objective value of LP  \eqref{lp:concave-closure} at $\hat{\mathbf{p}}$ --- denoted as $\Ex_{S \sim \hat{\mathbf{p}}} f(S) = \sum_{S\subseteq [n]} \hat{p}(S)  f(S)$ ---  is close to its optimal objective value  $f^+(a + 2\epsilon)$.  Note that the $(1 - \delta)$-approximability of $\hat{\mathbf{p}}$ is with respect to the optimal sender utility under \emph{exact} persuasiveness constraints.  In other words, we have 
$$\frac{1}{2} \big[ f([n]) + \Ex_{S \sim \hat{\mathbf{p}}} f(S) \big]  \geq \frac{1- \delta}{2} \big[f([n]) + f^+(a) \big]  ,$$
or equivalently, $\Ex_{S \sim \hat{\mathbf{p}}} f(S)  \geq  (1 - \delta)f^+(a)  - \delta f([n])$. Notice however, we want to show that $\Ex_{S \sim \hat{\mathbf{p}}} f(S)  $ is close to $f^+ (a+2\epsilon)$ while not $f^+(a)$. This gap is filled by the following lemma which shows that $f^+(a)$ is \emph{not} much smaller than  $f^+ (a+2\epsilon)$.

\begin{lemma}\label{lem:bound-closure}
	$ (1 + 2\epsilon/a) f^+(a) \geq f^+(a + 2\epsilon) $.
\end{lemma}
\begin{proof}
	Directly establishing the relationship between $f^+(a)$ and $f^+(a+2\epsilon)$ turns out to be difficult. We instead establish the relation between	$f^+(a)$ and $f^+(a - 2\epsilon)$, and then utilize the concavity of $f^+(x)$ to relate  $f^+(a), f^+(a - 2\epsilon)$ to $f^+(a + 2\epsilon)$. 
	
	We first prove that for any $a \in (0,1)$ and $\epsilon \in (0,\frac{a}{2})$, $f^+(a -2\epsilon) \geq f^+(a) \cdot (1 - 2\epsilon/a)$. Let $\{ p^*(S) \}_{S \subseteq [n]}$ be the distribution that achieves $f^+(a)$, i.e., the optimal solution to LP \eqref{lp:concave-closure} with $
	\epsilon = 0$.  We use $\{ p^*(S) \}_{S \subseteq [n]}$  to  construct a feasible (not necessarily optimal) solution to LP \eqref{lp:concave-closure}  with parameter $a - 2\epsilon$  instead of $a + 2 \epsilon$ in the first set of constraints. For any $S \not = \emptyset$, let $p(S) = \frac{a - 2\epsilon}{a} p^*(S)$ and $p(\emptyset)  = 1 - \sum_{S \not = \emptyset} p(S)$. This is feasible to LP \eqref{lp:concave-closure} with parameter $a - 2\epsilon$ because the probabilities for all the non-empty sets are scaled down by factor $(a - 2\epsilon)/a$. The objective for this new set of variable values is 
	\begin{eqnarray*}
		\sum_{S \subseteq [n]} p(S) f(S) &=& \sum_{S \subseteq [n]} \frac{a - 2\epsilon}{a} \cdot p^*(S)  f(S) + \frac{2\epsilon}{a} \cdot f(\emptyset) \\
		&\geq &   \frac{a - 2\epsilon}{a}  \sum_{S \subseteq [n]}  p^*(S)  f(S) = \frac{a - 2\epsilon}{a}  f^+(a). 
	\end{eqnarray*}
	
	Since $\{ p(S)\}_{S \subseteq [n]}$ is a feasible solution to LP \eqref{lp:concave-closure} with parameter $a-2\epsilon$, we thus have $f^+(a - 2\epsilon) \geq \sum_{S \subseteq [n]} p(S) f(S)  \geq \frac{a - 2\epsilon}{a}  f^+(a)$. Now, observe that $f^+(x)$ is a concave  function in $x \in (0,1)$. This is a well-known property of linear program maximization problem ( see, e.g., \cite{boyd2004convex}).  The concavity implies $f^+(a - 2\epsilon) + f^+(a + 2\epsilon) \leq 2 f^+(a)$. We thus have
	$$f^+(a + 2\epsilon) \leq 2 f^+(a) - f^+(a - 2\epsilon) \leq (1 + 2\epsilon/a) f^+(a).$$
\end{proof}

Invoking  $\delta$-optimality of $\hat{\mathbf{p}}$ and Lemma \ref{lem:bound-closure},  we have 
\begin{eqnarray*}
	\Ex_{S \sim \hat{\mathbf{p}}} f(S)   &\geq&  (1 - \delta)f^+(a)  - \delta f([n])  \\
	& \geq &  \frac{1 - \delta}{1 + 2\epsilon/a}f^+(a+2\epsilon) - \delta  \\
	& \geq & f^+(a+2\epsilon) - 2\epsilon/a  -  2 \delta  
\end{eqnarray*}

This implies that $\hat{\mathbf{p}}$ is an (additive) $(2\epsilon/a+2\delta)$  approximation to LP \eqref{lp:concave-closure}. Crucially, the loss $2\epsilon/a+2\delta$ is a constant since $a$ is also a constant in Lemma \ref{lem:submodular-PTAS-hard}. These overall establish the NP-hardness of designing a $\poly(n^{1/\epsilon}, |\Theta|) $ time algorithm for computing a $(1-\epsilon)$-optimal $\epsilon$-persuasive private signaling scheme.   

\subsection{Proof of Proposition \ref{thm:bicriteria-approx}}\label{sec:bicriteria:app-framework}

We say a distribution $\tilde{p} \in \Delta_{\Theta}$ is $K$-uniform if each of its entry $p_{\theta}$ is a multiple of $1/K$.  With slight abuse of notation, let $\Delta_K \subseteq \Delta_{\Theta}$ denote the set of all $K$-uniform distributions. Note that  $|\Delta_K| = \O( |\Theta|^K)$.  The main idea of the proof is to convert any optimal signaling scheme to an efficiently computable $\epsilon$-persuasive $\alpha(1 - \beta \delta)$-optimal signaling scheme whose posteriors are all $(2\ln(\frac{2}{\delta})/\epsilon^2)$-uniform distributions. For constant $\epsilon, \delta$, we can then search for such a signaling scheme via a polynomial-size linear program by focusing only on $K$-uniform posterior distributions. Details are presented in Algorithm \ref{alg:bicriterial}. 

	\SetKwInOut{Parameter}{Parameters} 	
	\begin{algorithm}[t] 
		\SetAlgoNoLine
		\Parameter{ A small constant $\epsilon > 0$} 
		\KwIn{ $\{ u_i(\theta) \}_{i \in [n], \theta \in \Theta}$, value oracle access to $f$, an $\alpha$-approximate subroutine for $f$}
		\KwOut{ A public signaling scheme} 
		Set $K = 2\ln(\frac{2}{\epsilon})/\epsilon^2$; Compute set $\Delta_K$, consisting of all $K$-uniform distributions \;
		\For{ $\tilde{p} \in \Delta_K $}{ 
			Compute  sets  $\tilde{A} = \{  i \in [n]:  \sum_{\theta} \tilde{p}_{\theta}\cdot u_i(\theta) > \epsilon \} $, $\tilde{B} = \{  i \in [n]:  \sum_{\theta} \tilde{p}_{\theta}\cdot u_i(\theta) < -\epsilon \} $ and $\tilde{C} = [n] \setminus (A \cup B)$ \; 
			Let $\mathbf{s}_{\tilde{A}}(\tilde{p})= \mathbf{1}$, $\mathbf{s}_{\tilde{B}}(\tilde{p}) = \mathbf{0}$; Employ the $\alpha$-approximate subroutine to compute $\mathbf{s}_{\tilde{C}}(\tilde{p})$ as an $\alpha$-approximation to $$\max_{\mathbf{x}_{\tilde{C}} \in \{ 0,1 \}^{\tilde{C}}} f(\mathbf{s}_{\tilde{A}}(\tilde{p}), \mathbf{s}_{\tilde{B}}(\tilde{p}), \mathbf{x}_{\tilde{C}}).$$
		}
		
		Solve the following linear program with variables $x(\tilde{p})$'s to obtain the public signaling scheme:  
		\begin{lp}\label{lp:bicriteria}
			\maxi{  \sum_{\tilde{p} \in \Delta_K} x(\tilde{p}) \cdot f(\mathbf{s}( {\tilde{p}}))  }
			\st
			\con{ \sum_{\tilde{p} \in \Delta_K } x(\tilde{p}) \cdot \tilde{p}  = \lambda }
			\con{ \sum_{\tilde{p} \in \Delta_K } x(\tilde{p}) = 1}
			\qcon{x(\tilde{p}) \geq 0}{\tilde{q} \in \Delta_K} 
		\end{lp}
		\caption{Bi-Criteria Approximation for Public Persuasion} 
		\label{alg:bicriterial}
	\end{algorithm}

To analyze Algorithm \ref{alg:bicriterial}, the main step is to prove that there always exists an $\epsilon$-persuasive $\alpha(1 - \beta \delta)$-optimal signaling scheme whose posteriors are all $K$-uniform distributions for $K =(2\ln(\frac{2}{\delta})/\epsilon^2)$.    We start by examining the sender's expected utility $U(p)$, as a function of any posterior distribution $p \in \Delta_{\Theta}$. In this proof, it will be convenient to view the sender's objective function $f$ as a function of $\mathbf{s} \in \{ 0,1 \}^n$. Let $\mathbf{x}$ denote the vector of best receiver responses to the posterior distribution $p$, which can be constructed as follows. Let $A = \{  i \in [n]:  \sum_{\theta} p_{\theta}\cdot u_i(\theta) > 0 \} $ denote the set of players whose unique best response  is action $1$ under posterior belief $p$, $B = \{  i \in [n]:  \sum_{\theta} p_{\theta}\cdot u_i(\theta) < 0 \} $ denote the set of players whose unique best response is action $0$, and $C = \{  i \in [n]:  \sum_{\theta} p_{\theta}\cdot u_i(\theta) = 0 \} $ is the set of receivers who are indifferent between action $0$ and $1$. So $\mathbf{x}_A = \mathbf{1}$, $\mathbf{x}_B = \mathbf{0}$ and we have the freedom to choose the actions for receivers in set $C$. Under optimality, we have 
\begin{eqnarray}\label{eq:best-set-def}
U(p) =f(\mathbf{x}_A, \mathbf{x}_B, \mathbf{x}_C), \text{ \quad where \quad  } \mathbf{x}_C = \argmax_{ \mathbf{s}_C \in \{ 0,1 \}^C}  f (\mathbf{x}_A, \mathbf{x}_B, \mathbf{s}_C). 
\end{eqnarray}

Let $K = 2\ln(\frac{2}{\delta})/\epsilon^2$ and $\mathcal{K}$ denote a set of $K$ i.i.d. samples of states from $p$. Let $\tilde{p}$ be the empirical distribution of the samples in $\mathcal{K}$; So $\tilde{p} \in \Delta_K$. Note that any player $i$'s expected utility under $\tilde{p}$ is a random variable $ \sum_{\theta} \tilde{ p_{\theta}} \cdot u_i(\theta) $ (depending on the samples), with mean $ \sum_{\theta} p_{\theta}\cdot u_i(\theta) $.  Let $\tilde{A} = \{  i \in [n]:  \sum_{\theta} \tilde{p}_{\theta}\cdot u_i(\theta) > \epsilon \} $ denote the set of players whose unique best response  is action $1$ under $\epsilon$-persuasiveness, $\tilde{B} = \{  i \in [n]:  \sum_{\theta} \tilde{p}_{\theta}\cdot u_i(\theta) < -\epsilon \} $, and $\tilde{C} = \{  i \in [n]:  \sum_{\theta} \tilde{p}_{\theta}\cdot u_i(\theta) \in [-\epsilon, \epsilon] \} $ is the set of receivers who are indifferent between action $0$ and $1$ under $\epsilon$-persuasiveness.   Note that $\tilde{p}$ and $\tilde{A}, \tilde{B}, \tilde{C}$ are also random, with randomness from the samples $\mathcal{K}$.  By standard concentration bound and our choice of $K$, we know that  $$\Pr \bigg( \bigg| \sum_{\theta} \tilde{ p_{\theta}} \cdot u_i(\theta) - \sum_{\theta} p_{\theta}\cdot u_i(\theta) \bigg| \leq   \epsilon \bigg)  \geq 1- \delta, \qquad  \forall i \in [n].$$

 Therefore, for any player $i \in C$, meaning $\sum_{\theta} p_{\theta}\cdot u_i(\theta) = 0$,  we know that $i \in \tilde{C}$ with probability at least $1 - \delta$.  For any $i \in A$, $i \in \tilde{A} \cup \tilde{C}$ with probability at least $1 - \delta$ and for any $i \in B$, $i \in \tilde{B} \cup \tilde{C}$ with probability at least $1 - \delta$. Let $\mathbf{s} \in \{ 0,1 \}^n$ denote the vector of best receiver responses under $\epsilon$-persuasiveness. We have  $\mathbf{s}_{\tilde{A}} = \mathbf{1}$, $\mathbf{s}_{\tilde{B}} = \mathbf{0}$ and the sender has the freedom to choose $\mathbf{s}_{\tilde{C}} \in \{ 0,1 \}_{\tilde{C}}$.   Now, for any fixed $\tilde{A}, \tilde{B}, \tilde{C}$, we define an auxiliary variable $\mathbf{y} \in \{ 0,1 \}^n$ based on the $\mathbf{x}$ described in Equation \eqref{eq:best-set-def}, as follows: 
 \begin{enumerate}
 	\item For any $i \in A$: if $i \in \tilde{A} \cup \tilde{C}$, let $y_i = x_i = 1$; otherwise $i \in \tilde{B}$ and let $y_i = 0$. 
 	\item For any $i \in B$: if $i \in \tilde{B} \cup \tilde{C}$, let $y_i = x_i = 0$;  otherwise $i \in \tilde{A}$ and let $y_i = 1$. 
 	\item for any $i \in C$: if $i \in \tilde{C}$, let $y_i = x_i$; if $i \in \tilde{A}$, let $y_i = 1$; if $i \in \tilde{B}$, let $y_i = 0$. 
 \end{enumerate}
 
Note that $\mathbf{y}$ is a random variable where the randomness comes from sets $\tilde{A}, \tilde{B}, \tilde{C}$.  By definition, $\mathbf{y}$ is a valid vector of receiver best responses under $\epsilon$-persuasiveness. Moreover, we have $y_i = x_i$ with probability at least $1 - \delta$ for any $i \in [n]$.  
  
 By $\alpha$-approximability of $f$, we can efficiently find $\mathbf{s}^*_{\tilde{C}} $ for any fixed $\tilde{A}, \tilde{B}, \tilde{C}$ such that 
 $$f(\mathbf{s}_{\tilde{A}}, \mathbf{s}_{\tilde{B}}, \mathbf{s}^*_{\tilde{C}}) \geq \alpha \cdot  \max_{\mathbf{s}_{\tilde{C} } \in \{ 0,1 \}^{\tilde{C}} }   f (\mathbf{s}_{\tilde{A}}, \mathbf{s}_{\tilde{B}}, \mathbf{s}_{\tilde{C}}) \geq \alpha \cdot   f (\mathbf{y}). $$ 
 where the last inequality is because $\mathbf{y}$ is one valid vector of receiver best responses. Taking expectation over the random samples $\mathcal{K}$ on both sides of the above inequality, we have
\begin{eqnarray*}\label{eq:noise-to-signal}
\Ex_{\mathcal{K}} \bigg[  f (\mathbf{s}_{\tilde{A}}, \mathbf{s}_{\tilde{B}}, \mathbf{s}^*_{\tilde{C}}) \bigg]   & \geq &  \alpha \cdot \Ex_{\mathcal{K}}  \bigg[ f(\mathbf{y})  \bigg]   \geq    \alpha \cdot (1 - \beta \delta)f(\mathbf{x})  =  \alpha (1 - \beta \delta) U(p) 
\end{eqnarray*}
where the second inequality is due to the $\beta$-noise stability of $f$ and the first equality is due to linearity of expectation and $\Ex_{\mathcal{K}}  (\tilde{p}_{\theta}) = p_{\theta}$.   

Note that $\Ex_{\mathcal{K}}  \tilde{p} = p$, i.e., any $p$ can be converted to the expectation of distributions in $\Delta_K$.  Moreover, the above derivation shows that such a  conversion decreases the sender's utility to at least its $\alpha(1- \beta\sigma)$ fraction, assuming $\epsilon$-persuasiveness.  Therefore, given the optimal signaling scheme, we can substitute each of its posterior distribution $p$ by a distribution over $\Delta_K$ as described above, and obtain a signaling scheme which is $\alpha(1 - \beta\delta)$-optimal and $\epsilon$-persuasive with $K =2\ln(\frac{2}{\delta})/\epsilon^2$. We can then search for this signaling scheme by solving  Linear Program \eqref{lp:bicriteria}.  


Note that the above proof also goes through  if both the $\alpha$-approximability of $f$ and $\beta$-stability are additive. That is, when Equation \eqref{eq:alpha-approx} is $$
f(\mathbf{s}^*_{T}, \mathbf{s}^0_{-T})  \geq  \max_{  \mathbf{s}_{T} \in \{ 0,1 \}^{T} } f(\mathbf{s}_{T}, \mathbf{s}^0_{-T}) - \alpha$$
and Equation \eqref{eq:def-stable} is 
\begin{equation*}
\Ex_{T \sim \mathbf{p}} f(T) \geq f(S) - \beta \epsilon,  
\end{equation*} 
the above argument yields an $(\alpha + \beta \sigma)$-optimal and $\epsilon$-persuasive public signaling scheme.

\subsection{Noise Stability of Other Set Functions}\label{append2:0}
In this subsection, we show that 	any set function $f: 2^{[n]} \to \RR_+$ is $n$-noise-stable.  Moreover, there exist submodular functions that are \emph{not} $\beta$-stable for any $\beta < n$. 

We first prove the $n$-noise-stability for any function $f$. Note that any $\epsilon$-noisy distribution $\mathbf{p}$ around $S$ must satisfy $p(S) \geq 1 - n \epsilon$, due to the union bound:  $p(S) \geq 1 - \sum_{i \in S} \Pr(i \not \in T) - \sum_{i \not \in S} \Pr(i \in T) \geq 1- n\epsilon$. This implies $\Ex_{T \sim \mathbf{p}} f(T) \geq p(S)f(S) \geq f(S) - n\epsilon$ for any $f : 2^{[n]} \to [0,1]$, as desired. 

Next we show that supermodular functions do not admit any non-trivial upper bound for noise stability.  Consider the supermodular function $f$ defined as follows: $f(T) = 1$ when $T = [n]$ and $f(T) = 0$ otherwise. We show that for any $\beta < n$, $f$ is not $\beta$-stable. Consider $\epsilon \leq 1/n$ and the following $\epsilon$-noisy distribution around  $S = [n]$: $p(T) = \epsilon$ for any $T$ such that $|T| = n-1$, $p([n]) = 1 - n \epsilon$ and    $p(T) = 0$ otherwise. For this $\epsilon$-noisy distribution, we have $\Ex_{T \sim \mathbf{p}} f(T) = f([n])\cdot p([n]) = 1 - n\epsilon =  f([n]) - n \epsilon$. Thus, $f$ cannot be $\beta$-stable for any $\beta < n$.

\newpage
\section{Proof of Theorem \ref{thm: pub-max-equiv}}\label{sec:relax:app}
We start by formulating the problem of computing the optimal cce-persuasive public scheme in \eqref{lp:optPub-cce}, an exponential-size linear program. The variable $\varphi(\theta, S)$ is the probability of recommending action $1$ to receivers in set $S \subseteq [n]$  at the state of nature $\theta$. 
\begin{lp}\label{lp:optPub-cce}
	\maxi{\sum_{\theta \in \Theta} \lambda(\theta) \sum_{S\subseteq [n]} \varphi(\theta,S)  f_{\theta}(S) }
	\st 
	\qcon{\sum_{\theta \in \Theta}\sum_{S:i\in S} \varphi(\theta,S) \lambda(\theta) x_{\theta,i} u_i(\theta) \geq C_i}{i = 1,...,n}
	\qcon{\sum_{S \subseteq [n]} \varphi(\theta,S) = 1}{\theta \in \Theta}
	\qcon{ \varphi(\theta,S) \geq 0}{\theta \in \Theta; S \subseteq [n]}
\end{lp}
Here $C_i = \max  \{ \sum_{\theta} \lambda_{\theta} u_i(\theta), 0 \} $ is  the optimal utility of receiver $i$ by best responding under the prior belief $\lambda$. The following lemma proves one direction of Theorem \ref{thm: pub-max-equiv}, namely, any efficient algorithm for maximizing $f \in \F$ minus a linear function  can be converted to an efficient algorithm for computing the optimal cce-persuasive public scheme.  

\begin{lemma}\label{lem:comb-to-persuasion}
	If there is a polynomial time algorithm that maximizes $f(S) - \sum_{i \in S} w_i$ for any $f \in \F$ and any weights $w_i  \in \RR$, then there is a polynomial time algorithm that computes the optimal cce-persuasive public scheme for any instance in  $\I(\F)$. 
\end{lemma} 
\begin{proof}
The proof examines the dual program of LP \eqref{lp:optPub-cce} and  shows that any algorithm for maximizing $f\in \F$ minus a linear function can be employed to construct a separation oracle for the dual. Specifically, the dual of LP \eqref{lp:optPub-cce} is the following LP with variables  $x_{\theta}$ for any $\theta \in \Theta$ and $y_i$ for any $i \in [n]$. 
	\begin{lp}\label{lp:CCBP_dual}
		\mini{ -\sum_{i=1}^n C_i \cdot y_i + \sum_{\theta \in \Theta } x_{\theta}   }
		\st
		\qcon{ x_{\theta} -  \lambda_{\theta} \sum_{i \in S}  u_i(\theta)  y_i \geq \lambda_{\theta} f_{\theta} (S) }{  \theta \in \Theta, S \subseteq [n] }
		\qcon{ y_i \geq 0}{i=1,...,n}
	\end{lp}
	
	Note that LP \eqref{lp:CCBP_dual} has polynomially many variables but exponentially many constraints. To solve LP \eqref{lp:CCBP_dual}, it suffices to design an efficient separation oracle for its feasible region, denoted as $\P$. For any given variable values $\{ \hat{x}_{\theta} \}_{\theta \in \Theta} \cup \{ \hat{y}_i \}_{i \in [n]}$, the second set of constraints of  LP \eqref{lp:CCBP_dual} (i.e., $y_i \geq 0$) is straightforward to check.  For the first set of constraints, for any $\theta \in \Theta$, we need to check whether  $\hat{x}_{\theta} \geq   \lambda_{\theta} \big[ \sum_{i \in S}^n  \hat{y}_i u_i(\theta) +  f_{\theta} (S) \big] $  holds for all $S \subseteq [n]$ or not. This  is equivalent to maximizing $F(s) = \sum_{i \in S}^n  \hat{y}_i u_i(\theta) +  f_{\theta} (S)$ over all $S \subseteq [n]$ and then check whether its optimal objective $OPT$ satisfies $\hat{x}_{\theta} \geq \lambda_{\theta} OPT$ or not.  If so, then constraint $x_{\theta} - \lambda_{\theta} \sum_{i \in S}^n  y_i u_i(\theta) \geq \lambda_{\theta} f_{\theta} (S) $  holds for any $S \subseteq [n]$.   Otherwise, the optimal solution $S^* = \arg \max_{S \subseteq [n]} F(S)$ corresponds to a separating hyperplane $ x_{\theta} -  \lambda_{\theta} \sum_{i\in S^*}   u_i(\theta)  y_i =  \lambda_{\theta} f_{\theta} (S^*) $ since  $\hat{x}_{\theta}  - \lambda_{\theta}  \sum_{i \in S^*}^n   u_i(\theta) \hat{y}_i <  \lambda_{\theta}  f_{\theta} (S^*) $ but any point in  $\P$ must satisfy $x_{\theta}  - \lambda_{\theta}  \sum_{i \in S^*}^n   u_i(\theta) y_i \geq \lambda_{\theta}  f_{\theta} (S^*) $ by definition.  
	
	As a result, any polynomial-time algorithm for maximizing any $f \in \F$ minus a linear function can be used to construct a polynomial time separation oracle for $\P$. By the computational equivalence between separation and optimization \cite{GLSbook}, we can solve LP \eqref{lp:CCBP_dual}, thus its dual LP \eqref{lp:optPub-cce}, in polynomial time.
	
\end{proof}

We now prove the converse, which is the more involved direction.  Namely, given any efficient algorithm for computing the optimal cce-persuasive public scheme for instances form $\I(\F)$, we design an efficient algorithm for the maximizing $f(S) -  \sum_{i \in S} w_i$ for any   $f \in \F$ and weights $\{w_i \}_{i \in [n]}$.  This is also where our proof diverges from that of \cite{Dughmi2017algorithmic} which only applies to the restricted case with monotone non-decreasing sender objectives.  

Given any $f \in \F $ and any weight $\{ \bar{w}_i \}_{ i \in [n]}$, we seek to maximize $\bar{f}(S) = f(S) - \sum_{i \in S} \bar{w}_i$ over $S \subseteq [n]$.  First, let $S_+ = \{ i: \bar{w}_i \geq 0 \}$ and $ S_- = \{ i: \bar{w}_i < 0 \}$ denote the set of indexes of non-negative [resp. negative] $\bar{w}_i$'s. Note that $S_+ \cup S_- = [n]$ and $S_+ \cap S_- = \emptyset$.  

The starting point of our reduction is  the following  linear program (a repetition of LP \eqref{lp:WM_Dual_0} stated here just for convenience), with linear coefficients $\{ \beta_i \}_{ i \in [n]}$ and $\alpha \in \RR$ as parameters, and $\{ w_i\}_{ i \in [n]}$ and $v \in \RR$ as variables:   
\begin{lp}\label{lp:WM_Dual}
	\mini{\sum_{ i \in  [n]} \beta_i \cdot w_i  + \alpha \cdot v} 
	\st
	\qcon{ v + \sum_{i \in S} w_i  \geq f(S) }{ S \subseteq [n]}  
	\qcon{w_i \geq 0}{i \in S_+}
	\qcon{w_i \leq 0 }{i \in S_-} 
\end{lp}

The first main step of our proof is to reduce maximizing $F$ to solving LP \eqref{lp:WM_Dual} for  all possible linear coefficients, as formally stated in the following lemma. 

\begin{lemma}\label{lem:WM_to_dual}
	Maximizing $\bar{f}(S)=f(S) - \sum_{i \in S} \bar{w}_i$ over $S \subseteq [n]$ reduces in polynomial time to solving LP \eqref{lp:WM_Dual} for all possible combinations of objective coefficients $\{ \beta_i \}_{ i \in [n]}$ and $\alpha \in \RR$. 
\end{lemma}
\begin{proof}
	First, we show that maximizing $\bar{f}(S) =f(S) - \sum_{i \in S} \bar{w}_i$ reduces in polynomial time to a separation oracle for the feasible region of LP \eqref{lp:WM_Dual}, denoted as $\P$ for convenience.  In particular, given any separation oracle for $\P$, we can check whether  $\{ \bar{w}_i \}_{i \in [n]} \cup \{ v \}$ is in $\P$ or not, for different $v\in \RR$.  Since $\{ \bar{w}_i \}_{i \in [n]}$ always satisfy the second and the third set of constraints by definition, so  $\{ \bar{w}_i \}_{i \in [n]} \cup \{ v \}$ is  feasible if and only if  $v \geq  f(S) - \sum_{i \in S} \bar{w}_i$ for all $S \subseteq [n]$, or equivalently, $v \geq  \max_{S \subseteq [n]}  [ f(S) - \sum_{i \in S} \bar{w}_i ]$. As a result, we can binary search for the $\bar{v}$ which makes  $\{ \bar{w}_i \}_{i \in [n]} \cup \{ \bar{v} \}$ almost feasible, but not quite. More precisely, let $B$ denote the bit complexity of $f(S)$ and $\bar{w}_i$, then the binary search returns the exact optimal objective of $F$ after $O(B)$ steps. By setting $\bar{v}$ equaling the optimal objective minus $2^{-B}$,   $\{ \bar{w}_i \}_{i \in [n]} \cup \{ \bar{v} \}$ will be infeasible and the returned separating hyperplane corresponds to the optimal solution to the problem of maximizing $F$. 
	
	Due to the polynomial time equivalence between separation and optimization \cite{GLSbook}, we know that maximizing $F$ also reduces to solving LP \eqref{lp:WM_Dual} for any parameters $\{ \beta_i \}_{ i \in [n]}$ and $\alpha \in \RR$, as desired. 
\end{proof}

We remark that the key difference between our reduction and that of \cite{Dughmi2017algorithmic} is in the construction of LP \eqref{lp:WM_Dual}. In particular, \cite{Dughmi2017algorithmic}  constructed a similar LP, but our LP \eqref{lp:WM_Dual} is more carefully constrained. This seems necessary for the case with general sender objectives. However, the challenge is that our more constrained LP is also arguably more difficult to solve. This is why our proof requires a more carefully crafted construction of the  persuasion instances in order to take care of the  intricacies arising from these additional constraints, as we will do next.  

In particular, the second main step of our reduction is to reduce solving LP \eqref{lp:WM_Dual} to computing the optimal cce-persuasive public scheme for instances in $\I (\F)$. We start by ruling out some situations where  LP \eqref{lp:WM_Dual} can be directly solved without employing persuasion problems as a subroutine, as stated in the following lemma. This will simplify our reduction later.    
\begin{lemma}\label{lem:WM_dual_simplification}
	The optimal solution to LP \eqref{lp:WM_Dual} can be computed in polynomial time when $\alpha \leq 0$.  
\end{lemma}
\begin{proof}
	The proof follows a case analysis. When $\alpha <0$, we can set $w_i = 0$ for all $i$ and let $v$ to be arbitrarily large. It is easy to verify that this variable assignment will always be feasible to LP \eqref{lp:WM_Dual} and leads to an objective of $\alpha v$, which tends to $-\infty$  as $v \to \infty$. 
	
	Next, we consider the case of $\alpha = 0$, with another level of case analysis. 
	\begin{enumerate}
		\item  If there exists $i \in S_+$ such that $\beta_i < 0$, we set $w_i = v/2$ and $w_j = 0 $ for all $j \not = i$. By letting $v \to \infty$, it is easy to verify that this variable assignment will be feasible to LP \eqref{lp:WM_Dual} and lead to objective value $\beta_i v/2$, again tending to $-\infty$ as $v \to \infty$. 
		\item Similarly,  if there exists $i \in S_-$ such that $\beta_i > 0$, we can set $w_i = -v/2$ and $w_j = 0 $ for all $j \not = i$, resulting in an objective of $-\infty$ as $v \to \infty$.
		\item Otherwise, we must have $\beta_i \geq 0$ for all $i \in S_+$ and $\beta_i \leq 0$ for all $i \in S_-$.  This implies that  the objective value equals $\sum_{i \in S_+} \beta_i w_i + \sum_{i \in S_-} \beta_i w_i + \alpha v \geq 0$ for any feasible solution. Now we can set $w_i = 0$ for all $i \in [n]$ and $v$ to be sufficiently large. This  will lead to a feasible and optimal solution to LP \eqref{lp:WM_Dual}, with optimal objective value $0$.  
	\end{enumerate}
	
	Therefore, when $\alpha \leq 0$, LP \eqref{lp:WM_Dual} can be directly solved in polynomial time. 
\end{proof}

As a result of Lemma  \ref{lem:WM_dual_simplification}, to solve LP \eqref{lp:WM_Dual} we only need to focus on the situation $\alpha>0$.  This case turns out to reduce to persuasion over $\I (\F)$, as formally stated in the following lemma. Note that, Lemma \ref{lem:VM_dual} and \ref{lem:WM_dual_simplification},  together with Lemma \ref{lem:WM_to_dual}, prove the other direction of Theorem \ref{thm: pub-max-equiv}.  

\begin{lemma}\label{lem:VM_dual}
	Solving LP \eqref{lp:WM_Primal} for $\alpha > 0$ and arbitrary $\{ \beta_i \}_{i=1}^n$ reduces in polynomial time to computing the optimal cce-persuasive public schemes for instances in $\I (\F)$. 
\end{lemma}
\begin{proof}
	Since $\alpha > 0$ and re-scaling all the coefficients of a linear program by a strictly positive factor will not change its optimal solution, we can rescale  coefficients of  LP \eqref{lp:WM_Dual} by the factor $1/\alpha$ to obtain an equivalent LP with $\alpha = 1$. We instead consider solving the  dual program of LP \eqref{lp:WM_Dual} (assuming $\alpha=1$), as follows: 
	\begin{lp}\label{lp:WM_Primal}
		\maxi{\sum_{ S \subseteq [n] } p(S) \cdot f(S)} 
		\st
		\con{\sum_{ S \subseteq [n]} p(S)  = 1 }
		\qcon{\sum_{S: i \in S} p(S)  \leq  \beta_i}{i \in S_+}
		\qcon{\sum_{S: i \in S} p(S)  \geq  \beta_i}{i \in S_-}
		\qcon{p(S) \geq 0}{ S \subseteq [n]}
	\end{lp} 
	where $ p(S)$ for all $S \subseteq [n]$ are variables. 
	
	We claim that, to solve LP \eqref{lp:WM_Primal}, it is without loss of generality to focus on the case $\beta_i \in [0,1]$ for any  $i \in [n]$. This is because the constraints $p(S) \geq 0$ and $\sum_{S \subseteq [n]} p(S) = 1$ imply $\sum_{S: i \in S} p(S) \in [0,1]$. Therefore, for any $i \in S_+$,  $\beta_i <0$ will imply infeasibility of LP \eqref{lp:WM_Primal} while $\beta_i > 1$ will lead to the same feasible space as $\beta_i = 1$. In other words, it is without loss to only consider $\beta_i \in [0,1]$ for $i \in S_+$. Similar argument shows that we can also w.l.o.g. focus on $\beta_i \in [0,1]$ for any $i \in S_-$.   
	
	Assuming $\beta_i \in [0,1]$ for all $ i \in [n]$, we now construct a persuasion instance and show that its optimal cce-persuasive public scheme  informs an  optimal solution to LP \eqref{lp:WM_Primal}. There are $n$ receivers and two states of nature $\theta_0, \theta_1$ with $\lambda(\theta_0) = \lambda(\theta_1) = 1/2$. The sender's utility function satisfies  $f_{\theta_0} \equiv 0$ (i.e., the trivial function) and   $f_{\theta_1} = f$ (i.e., the function of our interest). Define $u_i(\theta_0) = \beta_i$ and $u_i (\theta_1) = -1$ for any $i \in S_+$, while define $u_i(\theta_0) = -(1-\beta_i)$ and $u_i(\theta_1) = 1$  for any $i \in S_-$.   
	
	Let $\varphi^*$ be an optimal cce-persuasive public scheme, in particular an optimal solution to the instantiation of LP \eqref{lp:optPub-cce}  for our instance. Without loss of generality, we can adjust $\varphi^*$ so that at the state $\theta_0$ it always recommends action $1$ to any receiver $i \in S_+$ and recommends action $0$ to any receiver $i \in S_-$  (i.e., setting $\varphi^*(\theta_0, S_+) = 1$).  This will maintain the feasibility of $\varphi^*$  because the adjustment does not violate any cce-persuasiveness constraints  ---  it recommends each receiver's true optimal action at the state $\theta_0$ and this would only strengthen each receiver's incentive. Moreover, the adjustment maintains optimality of $\varphi^*$ because the sender's utility satisfies $f_{\theta_0} \equiv 0$ anyway. 
	
	We now instantiate LP \eqref{lp:optPub-cce} for our constructed instance, w.l.o.g., after restricting $\varphi^*(\theta_0, S_+) = 1 $. This leads to the following linear program, for which we know $\varphi^*$ is an optimal solution. 
	\begin{lp}\label{lp:optInstance}
		\maxi{\frac{1}{2} \cdot 0 + \frac{1}{2}  \sum_{S\subseteq [n]} \varphi(\theta_1,S)  f(S) }
		\st 
		\con{\sum_{S \subseteq [n]} \varphi(\theta_1,S) = 1}
		\qcon{1 \cdot (\beta_i) + \sum_{S:i\in S} \varphi(\theta_1,S) \cdot (-1) \geq  \max \{ \beta_i - 1, 0 \}}{i \in S_+}
		\qcon{ 0 \cdot (-1 + \beta_i) + \sum_{S:i\in S} \varphi(\theta_1,S) \cdot 1 \geq  \max \{ - 1+ \beta_i  +1, 0 \}}{i \in S_-}
		\qcon{ \varphi(\theta_1,S) \geq 0}{\theta \in \Theta; S \subseteq [n]}
	\end{lp}  
	After simplifications, it is easy to see that the second set of constraints is  precisely  $\sum_{S:i\in S} \varphi(\theta_1,S)  \leq \beta_i$ for any $i \in S_+$ and the third set of constraints is $\sum_{S:i\in S} \varphi(\theta_1,S)  \geq \beta_i$ for any $i \in S_-$. It is now clear that setting $p(S) = \varphi^*(\theta_1, S)$ yields an optimal solution to LP~\eqref{lp:WM_Primal}. 
\end{proof}

\end{document}